\documentclass[11pt]{article}
\usepackage{epsfig}
\usepackage{amssymb,amsmath,amsthm,url}
\usepackage{multirow}
\usepackage{booktabs}
\usepackage{natbib}
\usepackage{setspace}
\usepackage{dsfont}
\usepackage{amsmath}
\usepackage{amssymb}
\usepackage{amsthm}
\usepackage{upgreek}
\usepackage{mathrsfs}
\usepackage{makecell}
\usepackage{verbatim}
\usepackage[symbol]{footmisc}
\usepackage{lscape}
\usepackage{caption}
\usepackage{enumerate}
\usepackage{hyperref}
\usepackage{times,amsfonts,eucal}
\usepackage{times,amsfonts}
\usepackage{algorithm}
\usepackage{algorithmic}
\usepackage{graphicx,ifthen}
\usepackage{wrapfig}
\usepackage{latexsym}
\usepackage{color}
\usepackage{bm}
\usepackage{bbm}
\usepackage{srctex}
\usepackage[show]{simple_notes}

\newcommand{\ignore}[1]{}
\usepackage{bm}

\newtheorem{theorem}{Theorem}[section]
\newtheorem{lemma}{Lemma}[section]

\newtheorem{proposition}{Proposition}[section]

\newtheorem{remark}{Remark}[section]

\newtheorem{assumption}{Assumption}[section]
\newtheorem{algo}{Algorithm}[section]
\numberwithin{equation}{section}
\numberwithin{theorem}{section}
\numberwithin{lemma}{section}
\numberwithin{pro}{section}
\numberwithin{cor}{section}
\numberwithin{definition}{section}
\numberwithin{cons}{section}
\numberwithin{rem}{section}
\numberwithin{exa}{section}
\numberwithin{table}{section}
\numberwithin{figure}{section}
\numberwithin{algo}{section}

\newcommand{\C}{\mathbf{C}}

\newcommand{\ts}{\textstyle}
\newcommand{\ttop}{^{\top}}
\newcommand{\ve}{\varepsilon}
\newcommand{\var}{\text{var}}

\newcommand{\f}{\mathbf{f}}

\newcommand{\tI}{\text{\rm{I}}}
\newcommand{\tII}{\text{\rm{II}}}
\newcommand{\tIII}{\text{\rm{III}}}

\def\B{{\mathcal{B}}}

\def\C{{\mathbb C}}

\def\P{{\mathbb P}}

\def\R{{\mathbb R}}

\def\F{\mathcal{F}}

\def\P{\mathbb{P}}

\def\cals_+{{\cals_+}}

\def\cals{{\mathcal{S}}}

\def\tr{\mathrm{tr}}

\def\beq{\begin{equation}}
\def\eeq{\end{equation}}
\def\bals{\begin{align*}}
\def\eals{\end{align*}}

\def\bal{\begin{align}}
\def\eal{\end{align}}

\allowdisplaybreaks

\begin{document}

\title{Bootstrapping spectral statistics in high dimensions\footnote{This research was partially supported by NSF grants DMS 1305858, DMS 1407530, and DMS 1613218. Debashis Paul is thanked for helpful feedback.
}}

\author{
Miles Lopes\footnote{Department of Statistics, University of California, One Shields Avenue, Davis, CA 95616, USA,  emails: \tt{melopes,ablandino,aaue@ucdavis.edu}}
\and Andrew Blandino$^\dagger$
\and Alexander Aue$^\dagger$
}

\date{\today}
\maketitle

\begin{abstract}
\setlength{\baselineskip}{1.66em}

\maketitle

Statistics derived from the eigenvalues of sample covariance matrices are called spectral statistics, and they play a central role in multivariate testing. Although bootstrap methods are an established approach to approximating the laws of spectral statistics in low-dimensional problems, these methods are relatively unexplored in the high-dimensional setting. The aim of this paper is to focus on linear spectral statistics as a class of prototypes for developing a new bootstrap in high-dimensions --- and we refer to this method as the Spectral Bootstrap.  In essence, the method originates from the parametric bootstrap, and is motivated by the notion that, in high dimensions, it is difficult to obtain a non-parametric approximation to the full data-generating distribution. 
From a practical standpoint, the method is easy to use, and allows the user to circumvent the difficulties of complex asymptotic formulas for linear spectral statistics. In addition to proving  the consistency of the proposed method, we provide encouraging empirical results in a variety of settings. Lastly, and perhaps most interestingly, we show through simulations that the method can be applied successfully to statistics outside the class of linear spectral statistics, such as the largest sample eigenvalue and others.

\medskip 


\end{abstract}

\setlength{\baselineskip}{1.66em}

\section{Introduction}
\label{sec:intro}

This paper is concerned with developing a new method for bootstrapping statistics derived from the eigenvalues of high-dimensional sample covariance matrices ---  referred to as spectral statistics. With regard to the general problem of approximating the distributions of these statistics, random matrix theory and bootstrap methods offer two complementary approaches.  On one hand, random matrix theory makes it possible to understand certain statistics in fine-grained detail, with the help of specialized asymptotic formulas. On the other hand, bootstrap methods offer the prospect of a general-purpose approach that may handle a variety of problems in a streamlined way. 
In recent years, the two approaches have developed along different trajectories, and in comparison to the rapid advances in random matrix theory~\citep{Bai:Silverstein:2010,Paul:Aue:2014,Yao:Zheng:Bai:2015}, relatively little is known about the performance of bootstrap methods in the context of high-dimensional covariance matrices. 

Based on these considerations, it is of basic interest to identify classes of spectral statistics for which bootstrap methods can succeed in high dimensions. For this purpose, linear spectral statistics provide a fairly broad class that may be used as a testing ground. Moreover, linear spectral statistics are an attractive class of prototypes because they are closely related to many classical statistics in multivariate analysis \citep{Muirhead:2005}, and also, because their probabilistic theory is well-developed --- which facilitates the analysis of bootstrap methods.

Much of the modern work on the central limit theorem for linear spectral statistics in high dimensions was initiated in the pioneering  papers~\cite{Jonsson:1982} and \cite{Bai:Silverstein:2004}. In the case when data are generated by an underlying matrix of i.i.d.~random variables, these papers assume that the variables are either Gaussian or have kurtosis equal to 3. Subsequent papers, such as \cite{Pan:Zhou:2008}, \cite{Lytova:Pastur:2009}, \cite{Zheng:2012}, \cite{Wang:Yao:2013} and~\cite{Najim:Yao:2016}, have sought to relax the kurtosis condition at the expense of having to deal with additional non-vanishing higher-order terms that alter the form of the limit. In particular, the central limit theorems derived in these papers lead to intricate expressions for the limit laws of linear spectral statistics. Furthermore, if the kurtosis differs from 3, the existence of a limiting distribution is not assured without extra assumptions on the population eigenvectors --- which is elucidated in the papers~\cite{Pan:Zhou:2008} and \cite{Najim:Yao:2016}.

Although the asymptotic formulas for linear spectral statistics provide valuable insight, it is important to note  that numerically evaluating them can be fairly technical. Typically, this involves plugging parameter estimates into expressions involving complex derivatives, contour integrals, or multivariate polynomials of high degree, which entail non-trivial numerical issues, as discussed in the papers \cite{Rao:Mingo:Speicher:Edelman:2008} and \cite{Dobriban:2015}. By contrast, bootstrap methods have the ability to bypass many of these details, because the formulas will typically be evaluated implicitly by a sampling mechanism.
 Another related benefit is that if the settings of an application are updated, a bootstrap method may often be left unaltered,  whereas formula-based methods may be more sensitive to such changes.

In low dimensions, the bootstrap generally works for smooth functionals of the sample covariance matrix, with difficulties potentially arising in certain cases; for example, if population eigenvalues are tied. An overview of these settings is given in \cite{Hall:2009}, and remedies of various kinds have been proposed in \cite{Beran:1985}, \cite{Dumbgen:1993} and \cite{Hall:2009}, among others. In high-dimensions, there are few contributions to the literature on bootstrap procedures, and those available report mixed outcomes. For example, \cite{Pan:Gao:Yang:2014} briefly discuss a high-dimensional bootstrap method for constructing test statistics based on linear spectral statistics, but a different method is ultimately pursued in that work. Outside of the class of linear spectral statistics, the recent paper \cite{ElKaroui:Purdom:2016} considers both successes and failures of the standard non-parametric bootstrap in high dimensions. Specifically, it is proven that when the population covariance matrix is effectively low-rank, the non-parametric bootstrap can consistently approximate the joint distribution of a fixed set of the largest sample eigenvalues. However, it is also shown that the non-parametric bootstrap can fail to approximate the law of the largest sample eigenvalue when its population counterpart is not well separated from the bulk. Concerning the implementation of the non-parametric bootstrap, the recent paper~\cite{Fisher:2016} develops an efficient algorithm in the context of high-dimensional principal components analysis. Another computationally-oriented work is~\cite{Rao:Mingo:Speicher:Edelman:2008}, which deals with inference procedures based on tracial moments.

The primary methodological contribution of this paper is a bootstrap procedure for linear spectral statistics that is both user-friendly, and consistent, under certain assumptions. In light of the mentioned difficulties of the non-parametric bootstrap in high dimensions, it is natural to consider a different approach inspired by the parametric bootstrap. Specifically, our approach treats the population eigenvalues and kurtosis as the essential parameters for approximating the distributions of linear spectral statistics.
 Likewise, the proposed algorithm involves sampling bootstrap data from a proxy distribution that is parameterized by estimates of the eigenvalues and kurtosis. The approach taken here bears some similarity with the bootstrap method of \cite{Pan:Gao:Yang:2014}, since both may be viewed as relatives of the parametric bootstrap. However, there is no further overlap, as the bootstrap in \cite{Pan:Gao:Yang:2014} is intended to produce a specific type of test statistic, and is not designed to approximate the distributions of general linear spectral statistics, as pursued here. 

The main theoretical contribution of this paper is the verification of bootstrap consistency. To place this result in context, it is worth mentioning that, to the best of our knowledge, a bootstrap consistency result for general linear spectral statistics has not previously been available in high dimensions, even when the true covariance matrix is diagonal, or when the data are Gaussian. Nevertheless, the results here are embedded in a more general setting that allows for non-diagonal covariance matrices with sufficiently regular eigenvectors, as well as \smash{non-Gaussian} data. The proof synthesizes recent results on the central limit theorem for linear spectral statistics and spectrum estimation~\citep{Najim:Yao:2016,Ledoit:Wolf:2015}. Along the way, consistency is also established for a new kurtosis estimator that may be useful in other situations. The theoretical results are complemented by a simulation study for  several types of linear spectral statistics, which indicates that the proposed bootstrap has excellent performance in finite samples, even when the dimension is larger than the sample size. One of the most interesting aspects of the method is that it appears to extend well to various nonlinear spectral statistics, for which asymptotic formulas are more scarce. This fact is highlighted through experiments on three nonlinear spectral statistics: the largest sample eigenvalue, the sum of the top ten sample eigenvalues and the spectral gap statistic. Moreover, the proposed bootstrap leads to favorable results when applied to several classical sphericity tests, including some nonlinear ones.

\section{Setting and preliminaries}
\label{sec:setting_assumptions}


Our analysis is based on a standard framework for high-dimensional asymptotics, involving a set of $n$ samples in $\R^p$, where the dimension $p=p(n)$ grows at the same rate as $n$. The data-generating model is assumed to satisfy the following conditions, where the samples are represented as the rows of the data matrix $X\in\R^{n\times p}$.  Even though $X$ depends on $n$, this is generally suppressed, except in some technical arguments, and the same convention is applied to a number of other objects.

\begin{assumption}[Data-generating model]
\label{assumption:data}
\def~{\hphantom{0}} For each $n$, the population covariance matrix $\Sigma_n\in \R^{p\times p}$ is positive definite. As $n\to\infty$, the dimension satisfies $p=p(n)\to\infty$, such that $\gamma_n=p/n\to\gamma$ for some constant $\gamma\in (0,\infty)\setminus\{1\}$. For each $n$, the data matrix $X$ can be represented as $X=Z\Sigma_n^{1/2}$, where the matrix $Z\in\R^{n\times p}$
 is the upper-left $n\times p$ block of a doubly-infinite array of i.i.d.~random variables satisfying $E(Z_{11})=0$, $E(Z_{11}^2)=1$, $\kappa:=E(Z_{11}^4)>1$, and $E(Z_{11}^8)<\infty$.
\end{assumption}
 
\begin{remark}
\textup{The restriction $\gamma\neq 1$ is made for purely technical reasons, in order to use existing theory for spectrum estimation, as discussed in Section \ref{sec:method}. The proposed method can, however, still be implemented when $p=n$.}
\end{remark}

Define the sample covariance matrix $\hat\Sigma_n=n^{-1} X^\top X$, and denote its ordered eigenvalues by
\[
\lambda_1(\hat\Sigma_n)\geq\cdots\geq\lambda_p(\hat\Sigma_n).
\] 
In the high-dimensional setting, asymptotic results on the eigenvalues of $\hat\Sigma_n$ are often stated in terms of the empirical spectral distribution $\hat H_n$, defined through
\begin{equation}\label{EQ:ESD}
\hat H_n(\lambda)=p^{-1} \sum_{j=1}^p1\{\lambda_j(\hat\Sigma_n)\leq\lambda\}.
\end{equation}
Denote by $H_n$ the population counterpart of $\hat H_n$ in \eqref{EQ:ESD}, using the population eigenvalues $\lambda_j(\Sigma_n)$ in place of the sample eigenvalues $\lambda_j(\hat\Sigma_n)$. 

The class of linear spectral statistics associated with $\hat\Sigma_n$ consists of statistics of the form
\begin{equation*}\label{EQ:LSS}
T_n(f)=\int f(\lambda)d\hat H_n(\lambda)=p^{-1}\sum_{j=1}^pf\{\lambda_j(\hat\Sigma_n)\},
\end{equation*}
where $f$ is a sufficiently smooth real-valued function on an open interval $\mathcal{I}\subset\R$. It will be sufficient to assume that $f$ has $k$ continuous derivatives, denoted by~$f\in\mathscr{C}^k(\mathcal{I})$, where $k$ will be specified in the pertinent results.
  To specify $\mathcal{I}$ in detail, define an interval $[a,b]$ with endpoints\label{intervalcomments}
\begin{align*}
a &=(1-\sqrt{\gamma})_+^2 \ts\liminf_n \lambda_p(\Sigma_n),\\[0.1cm]
b &=(1+\sqrt{\gamma})^2\ts\limsup_n \lambda_1(\Sigma_n),
\end{align*}
where $x_+^2=(\max\{x,0\})^2$. Note that the boundedness of the interval $[a,b]$ will be implied by Assumption \ref{assumption:esd} below.  In turn, $\mathcal{I}$ is allowed to be any open interval containing $[a,b]$.
The reason for this choice of $\mathcal{I}$ is that  asymptotically, it is wide enough to contain all of the eigenvalues of $\hat\Sigma_n$. More precisely, with probability 1, every eigenvalue of $\hat\Sigma_n$ lies in $\mathcal{I}$ for all large $n$; see~\cite{Bai:Silverstein:1998,Bai:Silverstein:2004}. Lastly, when referring to the joint distribution of linear spectral statistics arising from several functions $\f=(f_1,\dots,f_m)$, the notation 
{$T_n(\f)=(T_n(f_1),\dots,T_n(f_m))$} is used, with $m$ being a fixed number that does not depend on $n$.

The next assumption details additional asymptotic requirements on the population spectrum. Convergence in distribution is denoted as $\Rightarrow$.

\begin{assumption}[Regularity of spectrum]
There is a limiting spectral distribution $H$, so that as $n\to\infty$,
\begin{equation}\label{EQ:LSD}
H_n\Rightarrow H,
\end{equation}
where the support of $H$ is a finite union of closed intervals, bounded away from zero and infinity. Furthermore, there is a fixed compact interval in $(0,\infty)$ containing the support of $H_n$~for all large $n$.
\label{assumption:esd}
\end{assumption}

\noindent The existence of the limit~\eqref{EQ:LSD} is standard for proofs relying on arguments from random matrix theory. Meanwhile, the assumed structure on the support of $H$ allows us to make use of existing theoretical guarantees for estimating the spectrum of $\Sigma_n$, based on the QuEST algorithm of \cite{Ledoit:Wolf:2016}, to be discussed later.

In addition to Assumption ~\ref{assumption:esd}, regularity of the population eigenvectors is needed to establish the consistency of the proposed method in the case of non-Gaussian data, when $\kappa\neq 3$. A similar assumption has also been used previously in~\citep[Theorem 1.4]{Pan:Zhou:2008} in order to ensure the existence of a limiting distribution for linear spectral statistics. To state the assumption, recall some standard terminology. For any distribution function $F$, define its associated Stieltjes transform as the function $z\mapsto \int\frac{1}{\lambda-z}dF(\lambda)$, where $z$ ranges over the set $\C\setminus \R$. A second object to introduce is the Mar\v{c}enko--Pastur map, which describes the limiting spectral distribution. Specifically, under conditions weaker than Assumptions~\ref{assumption:data} and~\ref{assumption:esd}, it is a classical fact that the limit $\hat{H}_n\Rightarrow \mathcal{F}(H,\gamma)$ holds with probability 1, where $\mathcal{F}(H,\gamma)$ is a distribution that only depends on $H$ and $\gamma$. The notation $H_{\gamma}=\F(H,\gamma)$ will be used as a shorthand. The object $\mathcal{F}(\cdot,\cdot)$ was termed the Mar\v{c}enko--Pastur map in~\cite{Dobriban:2016}. Additional background may be found in the references~\cite{Marcenko:Pastur:1967} and \cite{Bai:Silverstein:2010}.

\begin{assumption}[Regularity of eigenvectors]\label{assumption:eigenvec}
Let $z\in\C \setminus \R$, and let $t_n(z)$ be the Stieltjes transform of $\mathcal{F}(H_n,\gamma_n)$.
Also, let $\Sigma_n=U_n\Lambda_nU_n\ttop$ be the spectral decomposition for $\Sigma_n$, and for each $\ell\in\{1,2\}$, define the non-random diagonal matrix
\begin{equation}\label{EQ:GAMMADEF}
\Gamma_{n,\ell}(z)=\Lambda_n^{1/2}\Big[-zI_p +\{(1-\gamma_n)-z\gamma_n t_n(z)\}\Lambda_n\Big]^{\!-\ell}\Lambda_n^{1/2}.
\end{equation}
Then, for any fixed numbers $z_1,z_2\in\C\setminus\R$, and for each $\ell\in\{1,2\}$, the following limit holds as $n\to\infty$:
\begin{equation}\label{EIGVEC1}
 \frac 1p\sum_{j=1}^p \{U_n\Gamma_{n,\ell}(z_1)U_n\ttop \}_{jj} \{U_n\Gamma_{n,2}(z_2)U_n\ttop\}_{jj}
 =  \frac 1p\sum_{j=1}^p \{\Gamma_{n,\ell}(z_1)\}_{jj} \{\Gamma_{n,2}(z_2)\}_{jj}\ +o(1).
\end{equation}
\end{assumption}

Note that the invertibility of the middle factor in the definition of $\Gamma_{n,\ell}(z)$ holds in general for $z\in\C\setminus\R$, and this can be verified from the proof of Lemma~\ref{LEM:STIELTJES} in the supplement. To comment on the condition~\eqref{EIGVEC1}, it is clearly satisfied when $\Sigma_n$ is diagonal, but more importantly, it can also be satisfied when $\Sigma_n$ is non-diagonal. Specific examples are detailed in Propositions~\ref{pro:spiked} and~\ref{pro:rankk} below. In addition, the simulation results reported in Section \ref{sec:emp} include several examples of non-diagonal $\Sigma_n$, as well as some constructed from natural data. One further point to keep in mind is that Assumption~\ref{assumption:eigenvec} will not be necessary when $\kappa=3$.

\section{Method } \label{sec:method}

This section details the bootstrap algorithm. At a conceptual level, the approach is rooted in the notion that when $p$ and $n$ are of the same magnitude, it is typically difficult to obtain a non-parametric approximation to the full data-generating distribution. Nevertheless, if the statistic of interest only depends on a relatively small number of  parameters, it may still be feasible to estimate them, and then generate bootstrap data based on the estimated parameters. In this way, the proposed method is akin to the parametric bootstrap even though the model in Assumption~\ref{assumption:data} is non-parametric.

From a technical perspective, the starting point for this method is the fundamental central limit theorem for linear spectral statistics established by \cite{Bai:Silverstein:2004} for the case $\kappa=3$. Their result implies that, under Assumptions~\ref{assumption:data} and~\ref{assumption:esd}, the statistic $p[T_n(f)-E\{T_n(f)\}]$ converges in distribution to $N\{0,\sigma_f^2(H,\gamma)\}$, where the limiting variance $\sigma_f^2(H,\gamma)$ is completely determined by $f$, $H$, and $\gamma$. Consequently, when $\kappa=3$, the eigenvectors of $\Sigma_n$ have no asymptotic effect on a standardized linear spectral statistic. More recently, the advances made by~\citet[Theorem 2]{Najim:Yao:2016} and~\citet[Theorem 1.4]{Pan:Zhou:2008} indicate that this property extends to the case \smash{$\kappa\not=3$,} provided that the eigenvectors of $\Sigma_n$ are sufficiently regular in the sense prescribed by Assumption \ref{assumption:eigenvec}. Likewise, this observation motivates 
a parametric-type bootstrap --- which  involves estimating the parameters $\kappa$ and $H$, and then drawing bootstrap data from a distribution that is parameterized by these estimates. 
More concretely, since $H$ can be approximated in terms of the finite set of population eigenvalues $\lambda_1(\Sigma_n),\dots,\lambda_p(\Sigma_n)$, the bootstrap data will be generated using estimates of these eigenvalues.

\subsection{Bootstrap algorithm}

To introduce the resampling algorithm, again let $\Sigma_n=U_n\Lambda_n U_n^\top$ be the spectral decomposition for $\Sigma_n$. The bootstrap method relies on access to estimators of the spectrum $\Lambda_n$ and the kurtosis $\kappa$, which will be denoted by $\tilde\Lambda_n$ and $\hat\kappa_n$. For the sake of understanding the resampling algorithm, these estimators may for now be viewed as black boxes. Later on, specific methods for obtaining $\tilde\Lambda_n$ and $\hat\kappa_n$ will be introduced and their consistency properties established. Lastly, if $W$ is a scalar random variable, write $W\sim \text{Pearson}(\mu_1,\mu_2,\mu_3,\mu_4)$ to refer to a member of the Pearson system of distributions, which is parameterized by the first four moments $\mu_l=E(W^l)$ with $l=1,\dots,4$~\citep{Becker2017,Pearson:1895}.
\begin{remark}
\textup{If the first three moments satisfy $\mu_1=0$,\, $\mu_2=1$, $\mu_3=0$, then any value $\mu_4>1$ is permitted within the standard definition of the Pearson system. However, as a matter of completeness, the possibility $\mu_4=1$ is included by defining Pearson(0,1,0,1) as the two-point Rademacher distribution placing equal mass at $\pm 1$. This small detail ensures that the distribution Pearson$(0,1,0,\hat\kappa_n)$ makes sense for all possible realizations of the estimator $\hat\kappa_n$ defined below in line~\eqref{EQ:EST-KURT}.}
\end{remark}

\begin{algo}[Spectral Bootstrap]
\label{algorithm:bootstrap}
\begin{tabbing} \\
  \enspace For: $b=1$ to $b=B$ \\
   \qquad  Generate a random matrix $Z^*\in\mathbb{R}^{n\times p}$ with i.i.d.~entries drawn from Pearson$(0,1,0,\hat{\kappa}_n)$. \\
   \qquad  Compute the eigenvalues of the matrix $\hat{\Sigma}_n^*=\frac 1n\tilde\Lambda_n^{1/2}(Z^*)\ttop (Z^*)\tilde\Lambda_n^{1/2}$, denoted as $\hat \lambda_1^*,\ldots,\hat \lambda_p^*$. \\
   \qquad Compute the statistic $T_{n,b}^*(f)=\frac 1p\sum_{j=1}^pf(\hat \lambda_j^*)$. \\
   \enspace Output the empirical distribution of the values $T_{n,1}^*(f),\ldots,T_{n,B}^*(f)$.
\end{tabbing} 
\end{algo}

Although the algorithm is presented with a focus on linear spectral statistics, it can be easily adapted to any other type of spectral statistic by merely changing the third step. The performance of the  algorithm may thus be explored in a wide range of situations.
Regarding the task of generating the random matrix $Z^*$, the Pearson system is used only because it offers a convenient way to sample from a distribution with a specified set of moments. Apart from the ability to select the first four moments as $(0,1,0,\hat{\kappa}_n)$, the choice of the distribution is non-essential.

\subsection{Estimating the spectrum} 

To use Algorithm \ref{algorithm:bootstrap} in practice, specific estimators for $H$ and $\kappa$ have to be specified. With regard to the first task of spectrum estimation, this has been an active topic, and several methods are available in the literature~\citep{ElKaroui:2008,Mestre:2008,Rao:Mingo:Speicher:Edelman:2008, Bai:Chen:Yao:2010,Ledoit:Wolf:2015,Valiant:2016}. For the purposes of this paper, a slightly modified version of  the QuEST spectrum estimation method proposed by \citet{Ledoit:Wolf:2015,Ledoit:Wolf:2016} is used. However, the bootstrap procedure does not uniquely rely on QuEST, and any other spectrum estimation method is compatible with the results presented here, as long as it furnishes a weakly consistent estimator of $H$, as in Theorem~\ref{THM:CONSIST} below. In addition to consistency properties, another reason for choosing the QuEST method  is its user-friendly Matlab software~\citep{Ledoit:Wolf:2016}.

For the bootstrap procedure of Algorithm \ref{algorithm:bootstrap}, the QuEST algorithm is used in the following way. Let $\hat{\lambda}_{\text{Q},1},\ldots,\hat{\lambda}_{\text{Q},p}$ denote the estimates of $\lambda_1(\Sigma_n),\ldots,\lambda_p(\Sigma_n)$ output by QuEST, noting that these eigenvalue estimates are obtained as quantiles of the QuEST estimator for $H$.
However, instead of using these eigenvalue estimates directly, the proposed bootstrap uses
\begin{equation}\label{EQ:LAMBDAEST}
   \tilde\lambda_j=\min\big\{\hat{\lambda}_{\text{Q},j} \, , \, \hat{\lambda}_{\text{bound},n}\big\},
   \qquad j=1,\ldots,p,
\end{equation}
where $\hat\lambda_{\text{bound},n}=2\lambda_1(\hat{\Sigma}_n)$. Going forward, the notation $\tilde\Lambda_n=\text{diag}(\tilde\lambda_1,\dots,\tilde\lambda_p)$ will be used in several places. Applying the truncation~\eqref{EQ:LAMBDAEST} ensures that the top estimated eigenvalue $\tilde\lambda_1$ remains asymptotically bounded, which will be useful in proving that the bootstrap method is consistent.  Any fixed number greater than 1 could be used in place of $2$ in the definition of $\hat \lambda_{\text{bound},n}$. Also, the truncation will not affect estimation of the limiting distribution $H$, because it only affects a negligible fraction of top QuEST eigenvalues. In particular, the truncation does not affect the weak convergence of the distribution 
\begin{equation}\label{EQ:TILDEHDEF}
\tilde H_n(\lambda)=\ts p^{-1} \sum_{j=1}^p 1\{\tilde\lambda_j\leq \lambda\}.
\end{equation}
The consistency of $\tilde H_n$ is stated in Theorem~\ref{THM:CONSIST} later on.

\subsection{Estimating the kurtosis} 

It remains to construct an estimator for the kurtosis. This is done by considering an estimating equation for $\kappa$ arising from the variance of a quadratic form~\citep[eqn. 9.8.6,][]{Bai:Silverstein:2010}. Specifically, under the data-generating model in Assumption~\ref{assumption:data}, it is known that
\[
\kappa=3+\frac{\mathrm{var}(\|X_{1\cdot}\|_2^2)-2\|\Sigma_n\|^2_F}{\sum_{j=1}^p\sigma_{j}^4},
\]
where $\|\cdot\|_F$ denotes Frobenius norm, $X_{1\cdot}$ is the first row of $X$,  and $(\sigma_1^2,\dots,\sigma_p^2)=\text{diag}(\Sigma_{p})$. The importance of this equation is that it is possible to obtain ratio-consistent estimates of the three unknown parameters on the right-hand side, even when the dimension is high. Define 
\begin{equation*}\label{EQ:EST_KURT_0}
\tau_n=\|\Sigma_n\|_F^2, \qquad
\nu_n=\mathrm{var}(\|X_{1\cdot}\|^2_2) \qquad\mbox{and}\qquad
\omega_n=\ts\sum_{j=1}^p\sigma_{j}^4,
\end{equation*}
and note that these parameters tend to grow in magnitude as $p$ increases. Define corresponding estimators
\begin{align*}
\hat\tau_n&=\mathrm{tr}(\hat\Sigma_n^2)-\ts\frac 1n\mathrm{tr}(\hat\Sigma_n)^2, \\[0.2cm]
\hat\nu_n&=\ts\frac{1}{n-1}\sum_{i=1}^n\big(\|X_{i\cdot}\|_2^2-\frac 1n\sum_{i^\prime=1}^n\|X_{i^\prime\cdot}\|_2^2\big)^2, \\[0.2cm]
\hat\omega_n&=\ts\sum_{j=1}^p\big(\frac 1n\sum_{i=1}^nX_{ij}^2\big)^2.
\end{align*}
These give rise to the the kurtosis estimator
\begin{equation}\label{EQ:EST-KURT}
\hat\kappa_n=\max\Big( 3+\frac{\hat\nu_n-2\hat\tau_n}{\hat\omega_n},1\Big),
\end{equation}
whenever $\hat\omega_n\not=0$. In the exceptional case when $\hat\omega_n=0$, the estimator $\hat\kappa_n$ is arbitrarily defined to be 3, but this is unimportant from an asymptotic standpoint. Also note that the $\max\{\cdot,1\}$ function in the definition~\eqref{EQ:EST-KURT} enforces the basic inequality $\{E(Z_{11}^4) \}^{1/4}\geq \{E(Z_{11}^2)\}^{1/2}=1$.
To the best of our knowledge, a consistent estimate for $\kappa$ has not previously been established in the high-dimensional setting, although the estimation of related moment parameters has been studied, for instance, in~\cite{Bai:Saranadasa:1996} and~\cite{Rigollet:2015}. Outside the context of linear spectral statistics, the estimator $\hat\kappa_n$ may be independently useful as a diagnostic tool for checking whether or not data are approximately Gaussian. 

\section{Main results}\label{sec:main}

This section collects the asymptotic results, including the consistency of the spectrum and kurtosis estimators, and the consistency of the Spectral Bootstrap procedure. In addition, examples of covariance models are provided that guarantee regularity of eigenvectors. The first result pertains to the consistency of the estimators. Its proof as well as those of all other statements in this section are collected in the supplement.

\begin{theorem}[Consistency of estimators]\label{THM:CONSIST}
Suppose that Assumptions~\ref{assumption:data} and~\ref{assumption:esd} hold. Then, as $n\to\infty$, 
\begin{equation}\label{EQ:KAPPACONSIST}
\hat\kappa_n \xrightarrow{ \ \ } \kappa, 
\end{equation}
in probability, 
\begin{equation}\label{EQ:SPECTRUMCONSIST}
 \tilde H_n\Rightarrow H,
 \end{equation}
almost surely, and
\begin{equation}\label{EQ:SPECBOUND}
\ts\sup_{n}\lambda_1(\tilde\Lambda_n)<\infty ,
\end{equation}
almost surely.
\end{theorem}

The following propositions discuss specific settings that satisfy Assumption~\ref{assumption:eigenvec} on eigenvector regularity. 
The first example concerns the spiked covariance model introduced by \cite{Johnstone:2001}, which has received considerable attention in the literature. Confer \citep{Baik:Silverstein:2006,Paul:2007,Bai:Yao:2008} as well as~\citep{Bai:Yao:2012} for additional background.
An important feature of this example is that an arbitrary set of eigenvectors will satisfy Assumption~\ref{assumption:eigenvec} when the eigenvalues are spiked.

\begin{proposition}[Non-diagonal spiked covariance models]\label{pro:spiked} Suppose the eigenvalues of~$\Sigma_n$ are given by 
$$\Lambda_n=\text{\emph{diag}}\{ \lambda_1(\Sigma_n),\dots,\lambda_k(\Sigma_n),1,\dots,1\},$$
where $\lambda_1(\Sigma_n),\ldots,\lambda_k(\Sigma_n)>1$,  and $\sup_n \lambda_1(\Sigma_n)<\infty$. In addition, suppose $k=o(p)$. Then, for any $p\times p$ orthogonal matrix $U_n$, the matrix $\Sigma_n=U_n\Lambda_n U_n\ttop$ satisfies Assumption~\ref{assumption:eigenvec}. 
\end{proposition}

For the next example, recall that by definition, a matrix $\Pi$ is an orthogonal projection if it is symmetric and satisfies $\Pi^2=\Pi$. It follows that if $\Pi$ is an orthogonal projection, then any matrix of the form \smash{$U_n=I_p-2\Pi$} satisfies $U_n\ttop U_n=I_p$, and is hence orthogonal. In a sense, the following example is dual to the first example, since it shows that an arbitrary set of eigenvalues are allowed under Assumption~\ref{assumption:eigenvec} if the rank of the perturbing matrix $\Pi$ does not grow too quickly.

\begin{proposition}[Rank-$k$ perturbations with $k\to\infty$]\label{pro:rankk}
Suppose the eigenvectors of $\Sigma_n$ are given by
\smash{$U_n=I_p-2\Pi,$}
where $\Pi$ is a $p\times p$ orthogonal projection matrix with $\text{rank}(\Pi)=o(p)$. Then, for any diagonal matrix $\Lambda_n$ with non-negative entries, the matrix $\Sigma_n=U_n\Lambda_n U_n\ttop$ satisfies Assumption~\ref{assumption:eigenvec}.  
\end{proposition}
\noindent To consider the simplest case of the proposition, note that even a rank-1 perturbation $\Pi$ can produce a dense matrix $\Sigma_n$ in which most off-diagonal entries are non-zero. Namely, if $\mathbf{1}\in\R^p$ denotes the all-ones vector and if $\Pi=\ts \mathbf{1}\mathbf{1}\ttop / p$, then it follows that the off-diagonal $(i,j)$ entry of $\Sigma_n$ will be non-zero whenever \smash{$\ts \{ \lambda_i(\Sigma_n)+\lambda_j(\Sigma_n) \}/2$} is different from $\ts \tr(\Sigma_n)/p$.

\begin{remark}\normalfont
The previous examples may be of interest outside the scope of the current paper, since they illustrate some consequences of the central limit theorem derived by~\cite{Najim:Yao:2016}. Specifically, for the choices of $\Sigma_n$ identified in Propositions~\ref{pro:spiked} and~\ref{pro:rankk}, it follows from~\citep[Theorem 2]{Najim:Yao:2016} that even when $\kappa\neq 3$, a limiting distribution exists for the standardized statistic $p[T_n(f)-E\{T_n(f)\} ]$  --- and it seems that this was not previously known. In general, when $\kappa\neq 3$ and $\Sigma_n$ is non-diagonal, a limiting distribution is not guaranteed to exist.
\end{remark}

\subsection{Bootstrap consistency}\label{sec:bootconsist} All ingredients have been collected in order to state bootstrap consistency, which is expressed in terms of the L\'evy--Prohorov metric between probability distributions. Let $\mathcal{L}(U)$ denote the distribution of a random vector $U\in\R^m$, and let $\mathscr{B}$ denote the collection of Borel subsets of $\mathbb{R}^m$. For any $A\subset\mathbb{R}^m$ and any $\delta>0$, define the $\delta$-neighborhood  $A^\delta=\{x\in\mathbb{R}^m\colon \inf_{y\in A} \|x-y\|_2\leq\delta\}$. For any two random vectors $U$ and $V$ in $\R^m$, the LP metric between their probability distributions is defined by
\begin{equation*}\label{EQ:LP}
d_{\rm LP}(\mathcal{L}(U),\mathcal{L}(V))
=\inf\big\{\delta>0\colon\mathbb{P}(U\in A)\leq \mathbb{P}(V\in A^\delta)+\delta~\mbox{for all}~A\in\mathscr{B}\big\}.
\end{equation*}
The LP metric plays a basic role in comparing distributions because convergence with respect to $d_{\rm LP}$ is equivalent to weak convergence. As one more piece of notation, the expression $T_{n,1}^*(\f)$ is understood to represent a bootstrap sample constructed from Algorithm \ref{algorithm:bootstrap} with the estimators $\tilde\Lambda_n$ and $\hat{\kappa}_n$ defined through \eqref{EQ:LAMBDAEST} and \eqref{EQ:EST-KURT}.

\begin{theorem}[Consistency of the spectral bootstrap for linear spectral statistics]\label{THM:BOOTSTRAP_CONSIST}
Suppose that Assumptions~\ref{assumption:data} and~\ref{assumption:esd} hold, and that either $\kappa=3$ or Assumption~\ref{assumption:eigenvec} hold. Let $\f=(f_1,\dots,f_m)$ be fixed functions lying in $\mathscr{C}^3(\mathcal{I})$. Then, as $n\to\infty$,
\begin{equation}\label{EQ:MAINRESULT}
d_{\rm LP}\Big(\mathcal{L}\big(p\{T_n(\f)-\mathbb{E}[T_n(\f)]\}\big) \ , \ \mathcal{L}\big(p\{T_{n,1}^*(\f)-\mathbb{E}[T_{n,1}^*(\f) \mid X]\} \mid X\big)\Big)
\xrightarrow{ \ \ } 0,
\end{equation}
in probability.
\end{theorem}
Note that the result allows for the approximation of the joint distribution of several linear spectral statistics, which is of interest, since a variety of classical statistics can be written as a non-linear function of several linear spectral statistics~\citep[Sec.~3.2]{Dobriban:2016}. 
A second point to mention is that even if $\kappa\neq 3$ and Assumption~\ref{assumption:eigenvec} fails, then the limit~\eqref{EQ:MAINRESULT} may still remain approximately valid if $\kappa$ is close to 3. 
The reason is based on the fact that the mean and variance of $T_n(f)$ can be expressed in the form
\begin{align*}
E \{ T_n(f) \}&=\mu_{n,1}(f)+(\kappa-3)\mu_{n,2}(f),\\[0.3cm]
\var\{T_n(f)\}&=\nu_{n,2}(f)+(\kappa-3)\nu_{n,2}(f),
\end{align*}
for some terms $\mu_{n,1}(f)$ and $\nu_{n,1}(f)$ that only depend on $\Sigma_n$ through its eigenvalues, as well as some other terms $\mu_{n,2}(f)$ and $\nu_{n,2}(f)$ that may depend on the eigenvectors of $\Sigma_n$. See part 2 of Theorem 1 in~\cite{Najim:Yao:2016}, or formulas 1.19 and 1.20 in~\cite{Pan:Zhou:2008}. Hence, if $\kappa$ is close to 3, then the factor $(\kappa-3)$ will reduce the effect of the eigenvectors on $E\{T_n(f)\}$ and $\var\{T_n(f)\}$, which thus reduces the importance of Assumption~\ref{assumption:eigenvec}.

The proof of Theorem \ref{THM:BOOTSTRAP_CONSIST} is given in the supplement. At a high level, the proof leverages recent progress on the central limit theorem for linear spectral statistics~\citep{Najim:Yao:2016}, as well as a consistency guarantee for spectrum estimation~\citep{Ledoit:Wolf:2015}. In particular, there are two ingredients from~\cite{Najim:Yao:2016} that are helpful in analyzing the bootstrap in the high-dimensional setting. The first is the use of the LP metric for quantifying distributional approximation, which differs from previous formulations of the central limit theorem for linear spectral statistics that have usually been stated in terms of weak limits.
Secondly, \cite{Najim:Yao:2016} make use of the Helffer--Sj\"ostrand formula~\citep{Helffer:Sjostrand:1989}, which allows $T_n(f)$ to be analyzed with $f\in \mathscr{C}^3(\mathcal{I})$, as opposed to the more stringent smoothness assumptions placed on $f$ in previous works. This formula is also convenient to work with, since it allows any linear spectral statistic to be represented as a linear functional of the empirical Stieltjes transform $\ts \mathrm{tr}\{(\hat{\Sigma}_n -zI_p)^{-1} \} / p$, viewed as a process indexed by $z$. Hence, when comparing $T_n(f)$ with its bootstrap analogue, it is enough to compare the empirical Stiltjes transform with its bootstrap analogue, and in turn, these have the virtue of being approximable with Gaussian processes.

The statistic $T_n(\f)$ is a natural estimate of the parameter $\vartheta(\f)=(\vartheta_n(f_1),\dots,\vartheta_n(f_m))$, defined by
\begin{equation}\label{EQ:BIASMEAN}
\vartheta_n(f) =\int f(\lambda)dH_{n,\gamma_n}\!(\lambda)
\end{equation}
where $f:[0,\infty)\to\R$ is a given function, and $H_{n,\gamma_n}=\F(H_n,\gamma_n)$. From the existing literature on linear spectral statistics, it is known that the bias 
$$b_n(f)=E \{T_n(f)\}-\vartheta_n(f),$$ 
has a magnitude comparable to the standard deviation of $T_n(f)$. Indeed, for a suitable $f$, the rescaled bias $p b_n(f)$ is known to converge to a non-zero limit under Assumptions~\ref{assumption:data}--\ref{assumption:eigenvec}~\citep{Pan:Zhou:2008,Najim:Yao:2016}. For this reason, it is of interest to know if the bootstrap can consistently estimate the bias. The purpose of Theorem~ \ref{THM:BOOTSTRAP_BIAS} below is to answer this question in the affirmative.

To define the bootstrap estimate of bias, note that the analogue of $\vartheta_n(f)$ in the bootstrap world is given by 
$\tilde\vartheta_n(f)=\int f(\lambda) d\tilde H_{n,\gamma_n}(\lambda)$, where the integral is taken with respect to the distribution $\tilde H_{n,\gamma_n}:=\F(\tilde H_n,\gamma_n)$. Note that the value $\tilde\vartheta_n(f)$ is a deterministic function function of $\tilde H_n$ and $\gamma_n$, which can be computed with a variety of techniques, such as direct Monte-Carlo approximation, specialized algorithms~\citep{Jing:2010,Dobriban:2015}, or asymptotic formulas~\citep{Wang:2014}.
In turn, for a given vector of functions $\f$, the bootstrap estimate of $b_n(\f)$ is defined as the difference 
$$\hat b_n(\f)=E\{T_{n,1}^*(\f) \mid X \}-\tilde\vartheta_n(\f).$$
In addition to showing that $\hat b_n(\f)$ consistently estimates $b_n(\f)$, the following result also shows that the uncentered bootstrap distribution $\mathcal{L}\big[p\{T_{n,1}^*(\f)-\tilde \vartheta_n(\f)\} \mid X\big]$ consistently approximates $\mathcal{L}\big[p\{T_n(\f)-\vartheta_n(\f)\}\big]$.
 
\begin{theorem}[Consistency of bootstrap bias estimate]\label{THM:BOOTSTRAP_BIAS}
Suppose that Assumptions~\ref{assumption:data} and~\ref{assumption:esd} hold, and that either $\kappa=3$ or Assumption~\ref{assumption:eigenvec} hold. Let $\f=(f_1,\dots,f_m)$ be fixed functions lying in $\mathscr{C}^{18}(\mathcal{I})$.  Then, as $n\to\infty$,
\begin{equation}\label{EQ:BIASRESULT}
  p \{ b_n(\f)-\hat b_n(\f) \} \
\xrightarrow{ \  \ } \  0,
\end{equation}
in probability. Furthermore,
\begin{equation}\label{EQ:BIASMETRIC}
d_{\rm LP}\Big(\mathcal{L}\big(p\{T_n(\f)-\vartheta_n(\f)\}\big) \ , \ \mathcal{L}\big(p\{T_{n,1}^*(\f)-\tilde \vartheta_n(\f)\} \mid X\big)\Big)
\xrightarrow{ \  \ } 0, 
\end{equation}
in probability.
\end{theorem}

If Assumption~\ref{assumption:eigenvec} does not hold, it is possible that the limits~\eqref{EQ:BIASRESULT} and~\eqref{EQ:BIASMETRIC} may remain approximately valid if $\kappa$ is sufficiently close to 3, for reasons similar to those discussed with regard to Theorem~\ref{THM:BOOTSTRAP_CONSIST}. More specifically, the vector $b_n(\f)$ can be approximated by an expression of the form $b_n'(\f)+(\kappa-3)b_n''(\f)$, where $b_n'(\f)$ only depends on $\Sigma_n$ through its eigenvalues, but $b_n''(\f)$ may depend on the eigenvectors of $\Sigma_n$. Hence, if $(\kappa-3)$ is small, then the influence of the eigenvectors on $b_n(\f)$ will be reduced. Further details may be found in the supplement. Lastly, the requirement that the component functions of $\f$ lie in $\mathscr{C}^{18}(\mathcal{I})$ arises from technical considerations explained in Remark 4.4 of~\citet{Najim:Yao:2016}.

\section{Numerical experiments}
\label{sec:emp}

This section highlights the empirical performance of the proposed Spectral Bootstrap procedure in a variety of settings. The performance for three generic choices of linear spectral statistics is reported in Section \ref{subsec:emp:LSS}, while  Section \ref{subsec:emp:NLSS}  discusses how the bootstrap performs for a number of nonlinear spectral statistics that are not covered by the theory. Next, Sections \ref{subsec:emp:hypo}  and~\ref{subsec:realdata} show how the proposed method can be applied to some popular multivariate hypothesis tests, in the context of both synthetic and natural datasets. Lastly, the code used for the bootstrap algorithm can be found online at https://github.com/AndoBlando/LSS\_Bootstrap.

\subsection{Simulation settings}
\label{subsec:emp:setting}

Three types of non-diagonal covariance matrices $\Sigma_n$ were considered. In each case, data were generated according to $X=Z\Sigma_n^{1/2}$, where the entries of $Z\in\mathbb{R}^{n\times p}$ are i.i.d.\ random variables. The specifications for each type of covariance matrix, labeled (a), (b), and (c) are given below:

%
(a). Spiked covariance model: The eigenvalues were chosen as $\lambda_1(\Sigma_n)=\cdots=\lambda_{10}(\Sigma_n)=3$ and $\lambda_j(\Sigma_n)=1$ for $j=11,\ldots,p$. The eigenvector matrix $U_n$ for $\Sigma_n$ was generated from the uniform Haar distribution on the set of orthogonal matrices.
%

(b). Spread eigenvalues:  Substantial variation among the eigenvalues was introduced by the choice $\lambda_j(\Sigma_n)=j^{-1/2}$ for $j=1,\ldots,p$. This choice is of special interest, since it violates the condition that the bottom eigenvalue is bounded away from 0 as $p\to\infty$, which is commonly relied upon in random matrix theory. In addition, the eigenvectors of $\Sigma_n$ were generated as in the spiked case.
%

(c). Real data: The population matrix $\Sigma_n$ was constructed
with the help of the `DrivFace' dataset in the~\cite{UCI} repository.
After centering the rows and standardizing the columns, the rows were projected onto the first $p$ principal components, with $p=200,400$, or 600.
If the resulting transformed data matrix is denoted $\tilde{X}$, then the matrix $\Sigma_n=\tilde{X}^T\tilde{X} /n $ was used, for each choice of $p$, as a population covariance matrix for generating new data in the simulations.

With regard to the entries of the matrix $Z$, they were drawn from the following three distributions, and then standardized to have mean 0 and variance 1:

 (1). Gaussian, for which $\kappa=3$.
 
 (2). beta(6,6), for which $\kappa=2.6$.
 
 (3). Student $t$-distribution with 9 degrees of freedom, for which $\kappa=4.2$.\\
\noindent The beta distribution is an example of a platykurtic distribution, while the $t$-distribution is leptokurtic. This allows for a meaningful assessment of the bootstrap for various choices of kurtosis. 
 For each combination of settings (a)--(c) and (1)--(3), simulation results are reported for sample size $n=500$ and dimensions $p=200$, $400$ and $600$, leading to aspect ratios $\gamma_n=0.4$, $0.8$ and $1.2$, respectively.

\subsection{Simulations for linear spectral statistics}
\label{subsec:emp:LSS}

The Spectral Bootstrap's ability to approximate the distribution of $p\{T_n(f)-\vartheta_n(f)\}$ was studied with several choices of $f$, namely: $f(x)=x$, corresponding to $T_n(f)=\mathrm{tr}(\hat\Sigma_n)$, $f(x)=x^2$, corresponding to $T_n(f)=\mathrm{tr}(\hat\Sigma_n^2)=\|\hat\Sigma_n\|_F^2$, and $f(x)=\log x$, corresponding to $T_n(f)=\log\det(\hat\Sigma_n)$.

For each setting corresponding to (a)--(c) and (1)--(3), a set of 50,000 realizations of $X$ were generated, and for each one, the statistic $p(T_n(f)-\vartheta_n(f))$ was computed. From this set of 50,000 realizations, the sample mean, standard deviation, and 0.95 quantile were recorded. These three values are viewed as a proxy for ground truth, and are reported in the first row corresponding to each choice of $\gamma_n$ in the tables below. With regard to the centering constant $\vartheta_n(f)$, it was 
computed by direct Monte-Carlo approximation, by averaging 50 realizations of $\frac{1}{40n}\sum_{j=1}^{40p}f\{\lambda_j(\hat{\mathsf{\Sigma}})\}$ where $\hat{\mathsf{\Sigma}}=\frac{1}{40n}\mathsf{X}\ttop \mathsf{X}\in\R^{40p\times 40p}$, and the matrix  $\mathsf{X}\in\R^{40n\times 40p}$ was generated as $\mathsf{X}=\mathsf{Z}\mathsf{\Sigma}^{1/2}$, with  $\mathsf{Z}\in\R^{40n\times 40p}$ consisting of i.i.d.~samples from $\text{Pearson}(0,1,0,\kappa)$, and $\mathsf{\Sigma}= I_{40}\otimes \Sigma_n\in\R^{40p\times 40p}$. The approaches in~\citep{Jing:2010,Dobriban:2015} were also considered for approximating $\vartheta_n(f)$, but the direct Monte-Carlo approximation seemed to provide the most favorable results overall.

From the 50,000 realizations of $X$ described in the previous paragraph, the following procedure was applied to the first 1,000 such matrices. The bootstrap method, as in Algorithm~\ref{algorithm:bootstrap}, was used to obtain $B=500$ replicates, $p\{T_{n,1}^*(f)-\tilde\vartheta_n(f)\},\dots,  p\{T_{n,B}^*(f)-\tilde\vartheta_n(f)\}$. With these 500 values, the sample mean, standard deviation, and 0.95 quantile were recorded as estimates of the population counterparts. Hence, 1,000 bootstrap estimates were obtained for each parameter, since 1,000 realizations of $X$ were used. The quantity $\tilde\vartheta_n(f)$ was approximated by analogy with method used for $\vartheta_n(f)$, with the only differences being that $\hat\kappa_n$ was used in place of $\kappa$, and $\tilde\Lambda_n$ was used in place of $\Sigma_n$. 
In the tables below, the sample means of these 1,000 estimates are reported in the second row corresponding to each choice of $\gamma_n$, with the sample standard deviation in parentheses.

The central limit theorem for linear spectral statistics ensures that under Assumptions~\ref{assumption:data},~\ref{assumption:esd}, and~\ref{assumption:eigenvec}, there are limiting mean and variance parameters $\eta(f)$ and $v(f)$ such that \mbox{$p\{T_n(f)-\vartheta_n(f)\}\Rightarrow N\{\eta(f),v(f)\}$} as $n\to\infty$. Using formulas given in~\cite{Pan:Zhou:2008}, it is possible to estimate~$\eta(f)$ and $v(f)$ by replacing all asymptotic quantities with finite-sample analogues. 
With regard to $\eta(f)$, this approach leads to the estimate
$\hat\eta_n(f)=\hat\eta_{n,1}(f)+\hat\eta_{n,2}(f)$
where
\begin{equation}
\hat\eta_{n,1}(f)=\ts\frac{-1}{2\pi\sqrt{-1}}\displaystyle\int_{\mathcal{C}_1} f(z)\frac{\gamma_n\int \hat{\underline{m}}^3(z)t^2d\tilde H_n(t)/\{1+t\hat{\underline{m}}(z)\}^3}{[1-\gamma_n\int \hat{\underline{m}}^2(z)t^2dH(t)/\{1+t\hat{\underline{m}}(z)\}^2]^2}dz,
\end{equation}
\begin{equation}
\hat\eta_{n,2}(f)=\ts\frac{3-\hat\kappa_n}{2\pi\sqrt{-1}}\displaystyle\int_{\mathcal{C}_1} f(z)\frac{\gamma_n \int \hat{\underline{m}}^3(z) t^2 d\tilde H_n(t)/\{1+t\hat{\underline{m}}(z)\}^3}{1-\gamma_n\int\hat{\underline{m}}^2(z)t^2d\tilde H_n(t)/\{1+t\hat{\underline{m}}(z)\}^2}dz,
\end{equation}
and the function $\hat{\underline{m}}(z)$ is defined by
$$\hat{\underline{m}}(z)=-\ts\frac{1-\gamma_n}{z}+\gamma_n\int \ts\frac{1}{\lambda -z}d\hat H_n(\lambda).$$
Likewise, an estimate of $v(f)$ may be obtained as
%
%
$\hat v_n(f)=\hat v_{n,1}(f)+\hat v_{n,2}(f)$
where
\begin{equation}
\hat v_{n,1}(f)=\ts\frac{-1}{2\pi^2}\displaystyle \int_{\mathcal{C}_1}\int_{\mathcal{C}_2}\ts\frac{f(z_1)f(z_2)}{\{\hat{\underline{m}}(z_1)-\hat{\underline{m}}(z_1)\}^2}\frac{d}{dz_1}\hat{\underline{m}}(z_1)\frac{d}{dz_2}\hat{\underline{m}}(z_2)dz_1dz_2,
\end{equation}
and
\begin{equation}
\hat v_{n,2}(f)=\ts\frac{\gamma_n(3-\hat\kappa_n)}{4\pi^2}\displaystyle \int_{\mathcal{C}_1}\int_{\mathcal{C}_2}f(z_1)f(z_2)\ts\frac{d^2}{dz_1dz_2}\Big[\hat{\underline{m}}(z_1)\hat{\underline{m}}(z_2)\displaystyle\int\ts\frac{t^2 d\tilde H_n(t)}{\{\hat{\underline{m}}(z_1)t+1\} \{\hat{\underline{m}}(z_2)t+1\}}\Big]dz_1dz_2.
\end{equation}
In the integrals above, the contours $\mathcal{C}_1$ and $\mathcal{C}_2$ are disjoint, oriented in the positive direction in the complex plane, and  enclose the support of $\hat H_n$. 
The integrals were computed using the integral and integral2 functions in MATLAB.

For each of the 1,000 realizations of $X$ in the bootstrap computations, the quantities $\hat\eta_n(f)$ and $\hat v_n(f)$ were computed as above. In turn, the mean, standard deviation, and 0.95 quantile of the distribution $N\{\hat\eta_n(f),\hat v_n(f)\}$ were recorded as the formula-based estimates of the three population counterparts. In the tables,
the sample mean and standard deviation of the 1,000 estimates are reported in the third row corresponding to each choice of $\gamma_n$.

\begin{remark}\textup{The application of the previous formulas when $\kappa\neq 3$ is a novel aspect of the current paper, since this is not possible without the estimate $\hat\kappa_n$.}
\end{remark}

\begin{remark}\textup{To address computational cost, the formula-based estimates do not require repeated computations as the bootstrap does, but they can still incur a non-negligible cost for two reasons. First, the contour integrals should be computed to very high precision --- for otherwise the formula-based estimates can have high variance in some particular cases, as discussed in the simulation results below. Second, the formulas still require estimates of the population eigenvalues, which are obtained by solving a large optimization problem in the case of the QuEST method. In particular, this optimization problem does not lend itself to parallelization. On the other hand, the bootstrap replicates are trivial to compute in parallel after the eigenvalue estimates have been obtained. Hence, if the user works within a distributed computing environment, then the extra cost of bootstrap replication is not necessarily a bottleneck in comparison to the other computations.}
\end{remark}

The results corresponding to the Gaussian, beta, and $t$-distributions are given in Tables \ref{tab:CLT-Gaussian}, \ref{tab:CLT-beta}  and  \ref{tab:CLT-t9}, respectively. Overall, both the bootstrap and the formula-based estimates show very good agreement with the population values. Nevertheless, there are some advantages and disadvantages of the two approaches. With regard to the bootstrap estimates, their standard errors tend to be a bit larger than those of the formula-based estimates. However, this excess variance can be  reduced by increasing the number of bootstrap samples $B$. Next, observe that when $f(x)=\log(x)$, $\gamma_n=0.8$, and $\Sigma_n$ is obtained from the `real data' case (c), the formula-based estimates have very high variance --- which seems to be due to numerical instabilities arising from very small eigenvalues. Furthermore, this occurs for all three choices of the $Z_{ij}$ distribution, whereas the bootstrap is unaffected by this issue. In any case, these particular differences are relatively minor in comparison to the overall similarity of the results.

\begin{table}
\caption{Case (1): Standard Gaussian variables $Z_{ij}$. Summary statistics for the distribution \mbox{$\mathcal{L}[p\{T_n(f)-\vartheta_n(f)\} ]$} with (a) spiked (b) spread and (c) real data covariance matrices $\Sigma_n$, and various aspect ratios $\gamma_n$.}{
\scriptsize
\begin{tabular}{cc@{\quad}lll@{\quad}lll@{\quad}lll}
	&& \multicolumn{3}{c}{\!\!\!\!\!\!\!\!\!\!\!\!\!\!\!\!\!\!\!\!\!\!\!\! \makecell{~\\[-0.3cm]$p\{T_n(f)-\vartheta(f)\}$\\[0.0cm]  with $f(x)=x\phantom{\displaystyle\frac 1n}$}} & \multicolumn{3}{c}{\!\!\!\!\!\!\!\!\!\!\!\!\!\!\!\!\!\!\!\!\!\!\!\! \makecell{$p\{T_n(f)-\vartheta(f)\}$\\[0.2cm]  with  $f(x)=x^2$}} & \multicolumn{3}{c}{\!\!\!\!\!\!\!\!\!\!\!\!\!\!\!\!\!\!\!\!\!\!\!\! \makecell{$p\{T_n(f)-\vartheta(f)\}$\\[0.2cm]  with  $f(x)=\log(x)$}} \\[.7cm]
	$\Sigma_n$ & $\gamma_n$ & mean & sd & 95th & mean & sd & 95th & mean & sd & 95th \\[.1cm] 
(a) &0.4 & 0.00 & 1.06 & 1.74 & 0.51 & 4.99 & 8.82 & -0.24 & 1.01 & 1.42\\ 
& &  {0.00} ({0.09}) & {1.06} ({0.05}) & {1.74} ({0.14}) & {0.51} ({0.42}) & {4.99} ({0.28}) & {8.79} ({0.76}) & {-0.24} ({0.09}) & {1.01} ({0.04}) & {1.43}  ({0.12}) \\ 
& &  {0.00} ({0.00}) & {1.06} ({0.04}) & {1.74} ({0.07}) & {0.56} ({0.08}) & {4.97} ({0.22}) & {8.74} ({0.44}) & {-0.25} ({0.03}) & {1.01} ({0.03}) & {1.41} ({0.02}) \\[.1cm] 
&0.8 & 0.00 & 1.26 & 2.07 & 0.76 & 4.83 & 8.73 & -0.78 & 1.80 & 2.16\\ 
& &  {0.00} ({0.10}) & {1.27} ({0.06}) & {2.08} ({0.16}) & {0.76} ({0.40}) & {4.86} ({0.21}) & {8.76} ({0.65}) & {-0.76} ({0.15}) & {1.80} ({0.06}) & {2.20}  ({0.21}) \\ 
& &  {0.00} ({0.00}) & {1.27} ({0.04}) & {2.08} ({0.07}) & {0.81} ({0.10}) & {4.85} ({0.15}) & {8.78} ({0.34}) & {-0.81} ({0.05}) & {1.81} ({0.23}) & {2.17} ({0.38}) \\[.1cm] 
&1.2 & 0.00 & 1.65 & 2.71 & 1.32 & 9.00 & 16.16 &  & &  \\ 
& &  {0.00}  ({0.14}) &  {1.65}  ({0.07}) &  {2.72}  ({0.21}) &  {1.28}  ({0.75}) &  {9.04}  ({0.41}) &  {16.20}  ({1.26}) &  & &  \\ 
& &  {0.00}  ({0.00}) & {1.65}  ({0.05}) & {2.72}  ({0.09}) & {1.37}  ({0.18}) & {9.02}  ({0.30}) & {16.21}  ({0.65}) & & & \\[.1cm] 
(b) &0.4 & 0.00 & 0.15 & 0.25 & 0.01 & 0.18 & 0.31 & -0.24 & 1.01 & 1.42\\ 
& &  {0.00} ({0.01}) & {0.15} ({0.01}) & {0.25} ({0.02}) & {0.01} ({0.01}) & {0.18} ({0.02}) & {0.31} ({0.04}) & {-0.23} ({0.09}) & {1.01} ({0.05}) & {1.42}  ({0.12}) \\ 
& &  {0.00} ({0.00}) & {0.15} ({0.01}) & {0.25} ({0.01}) & {0.01} ({0.00}) & {0.18} ({0.02}) & {0.30} ({0.03}) & {-0.25} ({0.04}) & {1.01} ({0.04}) & {1.40} ({0.03}) \\[.1cm] 
&0.8 & 0.00 & 0.16 & 0.27 & 0.01 & 0.18 & 0.32 & -0.77 & 1.81 & 2.17\\ 
& &  {0.00} ({0.01}) & {0.16} ({0.01}) & {0.27} ({0.02}) & {0.01} ({0.02}) & {0.18} ({0.02}) & {0.32} ({0.04}) & {-0.75} ({0.17}) & {1.80} ({0.08}) & {2.21}  ({0.21}) \\ 
& &  {0.00} ({0.00}) & {0.16} ({0.01}) & {0.27} ({0.02}) & {0.01} ({0.00}) & {0.18} ({0.02}) & {0.31} ({0.03}) & {-0.80} ({0.09}) & {1.79} ({0.05}) & {2.14} ({0.01}) \\[.1cm] 
&1.2 & 0.00 & 0.17 & 0.27 & 0.01 & 0.19 & 0.34 &  & &  \\ 
& &  {0.00}  ({0.01}) &  {0.17}  ({0.01}) &  {0.27}  ({0.03}) &  {0.01}  ({0.02}) &  {0.19}  ({0.02}) &  {0.33}  ({0.04}) &  & &  \\ 
& &  {0.00}  ({0.00}) & {0.17}  ({0.01}) & {0.27}  ({0.02}) & {0.01}  ({0.00}) & {0.19}  ({0.02}) & {0.32}  ({0.03}) & & & \\[.1cm] 
(c) &0.4 & 0.00 & 0.07 & 0.12 & 0.00 & 0.13 & 0.23 & -0.25 & 1.01 & 1.41\\ 
& &  {0.00} ({0.01}) & {0.07} ({0.01}) & {0.12} ({0.01}) & {0.00} ({0.01}) & {0.13} ({0.02}) & {0.23} ({0.04}) & {-0.24} ({0.09}) & {1.01} ({0.05}) & {1.42}  ({0.12}) \\ 
& &  {0.00} ({0.00}) & {0.07} ({0.01}) & {0.12} ({0.01}) & {0.00} ({0.00}) & {0.13} ({0.02}) & {0.22} ({0.03}) & {-0.25} ({0.04}) & {0.73} ({0.13}) & {0.94} ({0.21}) \\[.1cm] 
&0.8 & 0.00 & 0.07 & 0.11 & 0.00 & 0.13 & 0.22 & -0.78 & 1.79 & 2.16\\ 
& &  {0.00} ({0.01}) & {0.07} ({0.01}) & {0.11} ({0.01}) & {0.00} ({0.01}) & {0.13} ({0.02}) & {0.22} ({0.04}) & {-0.75} ({0.16}) & {1.79} ({0.07}) & {2.18}  ({0.21}) \\ 
& &  {0.00} ({0.00}) & {0.07} ({0.01}) & {0.11} ({0.01}) & {0.00} ({0.00}) & {0.13} ({0.02}) & {0.21} ({0.03}) & {-0.82} ({0.08}) & {6.90} ({4.17}) & {10.53} ({6.86}) \\[.1cm] 
&1.2 & 0.00 & 0.07 & 0.11 & 0.00 & 0.13 & 0.22 &  & &  \\ 
& &  {0.00}  ({0.01}) &  {0.07}  ({0.01}) &  {0.11}  ({0.01}) &  {0.00}  ({0.01}) &  {0.13}  ({0.02}) &  {0.22}  ({0.04}) &  & &  \\ 
& &  {0.00}  ({0.00}) & {0.07}  ({0.01}) & {0.11}  ({0.01}) & {0.00}  ({0.00}) & {0.13}  ({0.02}) & {0.21}  ({0.03}) & & & \\[.1cm]
\end{tabular} 
}

\caption{
 For each value of $\gamma_n$, the three associated rows are labeled as follows. The first row corresponds to the population quantities. The second row corresponds to the mean and standard deviation (in parentheses) for the bootstrap estimates. The third row corresponds to the mean and standard deviation (in parentheses) for the formula-based estimates.
}
\label{tab:CLT-Gaussian}	
\end{table}

\begin{table} 
\caption{Case (2): Standardardized beta(6,6) variables $Z_{ij}$. Results are displayed as in Table \ref{tab:CLT-Gaussian}.}{
\scriptsize
\begin{tabular}{cc@{\quad}lll@{\quad}lll@{\quad}lll}
			&& \multicolumn{3}{c}{\!\!\!\!\!\!\!\!\!\!\!\!\!\!\!\!\!\!\!\!\!\!\!\! \makecell{~\\[-0.3cm]$p\{T_n(f)-\vartheta(f)\}$\\[0.0cm]  with $f(x)=x\phantom{\displaystyle\frac 1n}$}} & \multicolumn{3}{c}{\!\!\!\!\!\!\!\!\!\!\!\!\!\!\!\!\!\!\!\!\!\!\!\! \makecell{$p\{T_n(f)-\vartheta(f)\}$\\[0.2cm]  with  $f(x)=x^2$}} & \multicolumn{3}{c}{\!\!\!\!\!\!\!\!\!\!\!\!\!\!\!\!\!\!\!\!\!\!\!\! \makecell{$p\{T_n(f)-\vartheta(f)\}$\\[0.2cm]  with  $f(x)=\log(x)$}} \\[.7cm]
			$\Sigma_n$ & $\gamma_n$ & mean & sd & 95th & mean & sd & 95th & mean & sd & 95th \\[.1cm]
(a) &0.4 & -0.01 & 0.97 & 1.59 & 0.33 & 4.73 & 8.20 & -0.18 & 0.93 & 1.35\\ 
& &  {0.01} ({0.07}) & {0.95} ({0.05}) & {1.56} ({0.12}) & {0.34} ({0.36}) & {4.48} ({0.24}) & {7.78} ({0.66}) & {-0.16} ({0.08}) & {0.93} ({0.04}) & {1.37}  ({0.11}) \\ 
& &  {0.00} ({0.00}) & {0.95} ({0.04}) & {1.56} ({0.06}) & {0.34} ({0.07}) & {4.48} ({0.20}) & {7.70} ({0.39}) & {-0.18} ({0.02}) & {0.93} ({0.03}) & {1.35} ({0.02}) \\[.1cm] 
&0.8 & 0.01 & 1.25 & 2.06 & 0.63 & 6.36 & 11.12 & -0.62 & 1.70 & 2.16\\ 
& &  {0.00} ({0.10}) & {1.24} ({0.06}) & {2.04} ({0.16}) & {0.55} ({0.51}) & {6.16} ({0.30}) & {10.74} ({0.85}) & {-0.60} ({0.14}) & {1.70} ({0.06}) & {2.20}  ({0.19}) \\ 
& &  {0.00} ({0.00}) & {1.24} ({0.05}) & {2.04} ({0.07}) & {0.58} ({0.11}) & {6.14} ({0.24}) & {10.68} ({0.49}) & {-0.42} ({0.38}) & {1.70} ({0.03}) & {2.37} ({0.37}) \\[.1cm] 
&1.2 & 0.00 & 1.48 & 2.45 & 0.82 & 8.29 & 14.46 &  & &  \\ 
& &  {0.00}  ({0.12}) &  {1.47}  ({0.07}) &  {2.43}  ({0.19}) &  {0.77}  ({0.67}) &  {8.14}  ({0.37}) &  {14.19}  ({1.12}) &  & &  \\ 
& &  {0.00}  ({0.00}) & {1.48}  ({0.05}) & {2.43}  ({0.08}) & {0.82}  ({0.15}) & {8.13}  ({0.27}) & {14.20}  ({0.59}) & & & \\[.1cm] 
(b) &0.4 & 0.00 & 0.14 & 0.24 & 0.01 & 0.17 & 0.31 & -0.15 & 0.93 & 1.33\\ 
& &  {0.00} ({0.01}) & {0.14} ({0.01}) & {0.23} ({0.02}) & {0.01} ({0.01}) & {0.16} ({0.02}) & {0.28} ({0.03}) & {-0.16} ({0.08}) & {0.93} ({0.05}) & {1.36}  ({0.11}) \\ 
& &  {0.00} ({0.00}) & {0.14} ({0.01}) & {0.23} ({0.01}) & {0.01} ({0.00}) & {0.16} ({0.02}) & {0.27} ({0.03}) & {-0.17} ({0.04}) & {0.92} ({0.04}) & {1.35} ({0.03}) \\[.1cm] 
&0.8 & 0.00 & 0.15 & 0.26 & 0.01 & 0.18 & 0.32 & -0.73 & 1.72 & 2.35\\ 
& &  {0.00} ({0.01}) & {0.14} ({0.01}) & {0.24} ({0.02}) & {0.01} ({0.01}) & {0.17} ({0.02}) & {0.28} ({0.04}) & {-0.60} ({0.15}) & {1.70} ({0.07}) & {2.20}  ({0.19}) \\ 
& &  {0.00} ({0.00}) & {0.14} ({0.01}) & {0.24} ({0.02}) & {0.01} ({0.00}) & {0.16} ({0.02}) & {0.28} ({0.03}) & {-0.64} ({0.08}) & {1.96} ({0.12}) & {2.58} ({0.19}) \\[.1cm] 
&1.2 & 0.00 & 0.16 & 0.26 & 0.01 & 0.19 & 0.33 &  & &  \\ 
& &  {0.00}  ({0.01}) &  {0.15}  ({0.01}) &  {0.25}  ({0.02}) &  {0.01}  ({0.01}) &  {0.17}  ({0.02}) &  {0.29}  ({0.04}) &  & &  \\ 
& &  {0.00}  ({0.00}) & {0.15}  ({0.01}) & {0.25}  ({0.02}) & {0.01}  ({0.00}) & {0.17}  ({0.02}) & {0.29}  ({0.03}) & & & \\[.1cm] 
(c) &0.4 & 0.00 & 0.06 & 0.11 & 0.00 & 0.12 & 0.20 & -0.17 & 0.93 & 1.33\\ 
& &  {0.00} ({0.01}) & {0.06} ({0.00}) & {0.11} ({0.01}) & {0.00} ({0.01}) & {0.12} ({0.01}) & {0.20} ({0.03}) & {-0.17} ({0.08}) & {0.93} ({0.04}) & {1.36}  ({0.11}) \\ 
& &  {0.00} ({0.00}) & {0.06} ({0.00}) & {0.11} ({0.01}) & {0.00} ({0.00}) & {0.12} ({0.01}) & {0.19} ({0.02}) & {-0.18} ({0.02}) & {2.05} ({0.91}) & {3.19} ({1.50}) \\[.1cm] 
&0.8 & 0.00 & 0.06 & 0.11 & 0.01 & 0.12 & 0.22 & -0.73 & 1.72 & 2.35\\ 
& &  {0.00} ({0.00}) & {0.06} ({0.00}) & {0.10} ({0.01}) & {0.00} ({0.01}) & {0.11} ({0.01}) & {0.20} ({0.03}) & {-0.61} ({0.14}) & {1.71} ({0.06}) & {2.20}  ({0.20}) \\ 
& &  {0.00} ({0.00}) & {0.06} ({0.00}) & {0.10} ({0.01}) & {0.00} ({0.00}) & {0.11} ({0.01}) & {0.19} ({0.02}) & {-0.66} ({0.05}) & {7.22} ({5.03}) & {11.22} ({8.27}) \\[.1cm] 
&1.2 & 0.00 & 0.06 & 0.10 & 0.00 & 0.11 & 0.19 &  & &  \\ 
& &  {0.00}  ({0.00}) &  {0.06}  ({0.00}) &  {0.10}  ({0.01}) &  {0.00}  ({0.01}) &  {0.12}  ({0.01}) &  {0.20}  ({0.03}) &  & &  \\ 
& &  {0.00}  ({0.00}) & {0.06}  ({0.00}) & {0.10}  ({0.01}) & {0.00}  ({0.00}) & {0.11}  ({0.01}) & {0.19}  ({0.02}) & & & \\[.1cm] 
		\end{tabular}
	}
\label{tab:CLT-beta}
\end{table}

\begin{table}
\caption{Case (3): Standard t-9 variables $Z_{ij}$. Results are displayed as in Table  \ref{tab:CLT-Gaussian}.}{
\scriptsize
		\begin{tabular}{cc@{\quad}lll@{\quad}lll@{\quad}lll}
			&& \multicolumn{3}{c}{\!\!\!\!\!\!\!\!\!\!\!\!\!\!\!\!\!\!\!\!\!\!\!\! \makecell{~\\[-0.3cm]$p\{T_n(f)-\vartheta(f)\}$\\[0.0cm]  with $f(x)=x\phantom{\displaystyle\frac 1n}$}} & \multicolumn{3}{c}{\!\!\!\!\!\!\!\!\!\!\!\!\!\!\!\!\!\!\!\!\!\!\!\! \makecell{$p\{T_n(f)-\vartheta(f)\})$\\[0.2cm]  with  $f(x)=x^2$}} & \multicolumn{3}{c}{\!\!\!\!\!\!\!\!\!\!\!\!\!\!\!\!\!\!\!\!\!\!\!\! \makecell{$p\{T_n(f)-\vartheta(f)\}$\\[0.2cm]  with  $f(x)=\log(x)$}} \\[.7cm]
			$\Sigma_n$ & $\gamma_n$ & mean & sd & 95th & mean & sd & 95th & mean & sd & 95th \\[.1cm]
(a) &0.4 & 0.00 & 1.31 & 2.18 & 1.08 & 5.67 & 10.52 & -0.47 & 1.23 & 1.55\\ 
& &  {0.00} ({0.11}) & {1.33} ({0.06}) & {2.21} ({0.18}) & {1.14} ({0.52}) & {6.30} ({0.36}) & {11.64} ({0.98}) & {-0.45} ({0.11}) & {1.22} ({0.05}) & {1.55}  ({0.14}) \\ 
& &  {0.00} ({0.00}) & {1.33} ({0.05}) & {2.19} ({0.09}) & {1.22} ({0.14}) & {6.23} ({0.30}) & {11.47} ({0.61}) & {-0.49} ({0.05}) & {1.22} ({0.04}) & {1.52} ({0.02}) \\[.1cm] 
&0.8 & 0.00 & 1.73 & 2.84 & 1.97 & 7.96 & 15.03 & -1.25 & 2.06 & 2.16\\ 
& &  {0.00} ({0.14}) & {1.75} ({0.09}) & {2.88} ({0.22}) & {1.97} ({0.74}) & {8.55} ({0.46}) & {16.15} ({1.31}) & {-1.19} ({0.18}) & {2.04} ({0.08}) & {2.16}  ({0.22}) \\ 
& &  {0.00} ({0.00}) & {1.75} ({0.06}) & {2.88} ({0.10}) & {2.11} ({0.22}) & {8.49} ({0.35}) & {16.08} ({0.77}) & {-1.28} ({0.09}) & {2.04} ({0.04}) & {2.08} ({0.02}) \\[.1cm] 
&1.2 & 0.00 & 2.06 & 3.39 & 2.80 & 10.68 & 20.43 &  & &  \\ 
& &  {-0.01}  ({0.17}) &  {2.08}  ({0.10}) &  {3.41}  ({0.27}) &  {2.74}  ({0.96}) &  {11.24}  ({0.55}) &  {21.36}  ({1.65}) &  & &  \\ 
& &  {0.00}  ({0.00}) & {2.08}  ({0.07}) & {3.42}  ({0.12}) & {2.96}  ({0.29}) & {11.18}  ({0.39}) & {21.35}  ({0.92}) & & & \\[.1cm] 
(b) &0.4 & 0.00 & 0.18 & 0.29 & 0.02 & 0.19 & 0.34 & -0.46 & 1.22 & 1.53\\ 
& &  {0.00} ({0.02}) & {0.19} ({0.01}) & {0.32} ({0.03}) & {0.02} ({0.02}) & {0.23} ({0.02}) & {0.41} ({0.05}) & {-0.46} ({0.11}) & {1.22} ({0.06}) & {1.55}  ({0.14}) \\ 
& &  {0.00} ({0.00}) & {0.19} ({0.01}) & {0.32} ({0.02}) & {0.03} ({0.00}) & {0.22} ({0.02}) & {0.40} ({0.04}) & {-0.49} ({0.06}) & {1.22} ({0.05}) & {1.52} ({0.02}) \\[.1cm] 
&0.8 & 0.00 & 0.19 & 0.30 & 0.02 & 0.19 & 0.34 & -1.25 & 2.03 & 2.06\\ 
& &  {0.00} ({0.02}) & {0.20} ({0.01}) & {0.34} ({0.03}) & {0.03} ({0.02}) & {0.24} ({0.02}) & {0.43} ({0.05}) & {-1.19} ({0.19}) & {2.04} ({0.09}) & {2.16}  ({0.23}) \\ 
& &  {0.00} ({0.00}) & {0.21} ({0.01}) & {0.34} ({0.02}) & {0.03} ({0.00}) & {0.23} ({0.02}) & {0.41} ({0.04}) & {-1.28} ({0.13}) & {2.04} ({0.06}) & {2.07} ({0.03}) \\[.1cm] 
&1.2 & 0.00 & 0.19 & 0.32 & 0.02 & 0.20 & 0.35 &  & &  \\ 
& &  {0.00}  ({0.02}) &  {0.21}  ({0.01}) &  {0.35}  ({0.03}) &  {0.03}  ({0.02}) &  {0.24}  ({0.02}) &  {0.44}  ({0.05}) &  & &  \\ 
& &  {0.00}  ({0.00}) & {0.21}  ({0.01}) & {0.35}  ({0.02}) & {0.03}  ({0.00}) & {0.24}  ({0.02}) & {0.42}  ({0.04}) & & & \\[.1cm] 
(c) &0.4 & 0.00 & 0.09 & 0.16 & 0.01 & 0.16 & 0.30 & -0.47 & 1.22 & 1.55\\ 
& &  {0.00} ({0.01}) & {0.09} ({0.01}) & {0.15} ({0.02}) & {0.01} ({0.01}) & {0.16} ({0.04}) & {0.29} ({0.07}) & {-0.44} ({0.15}) & {1.20} ({0.10}) & {1.53}  ({0.14}) \\ 
& &  {0.00} ({0.00}) & {0.09} ({0.01}) & {0.15} ({0.02}) & {0.01} ({0.00}) & {0.16} ({0.03}) & {0.27} ({0.06}) & {-0.48} ({0.14}) & {1.00} ({0.15}) & {1.17} ({0.16}) \\[.1cm] 
&0.8 & 0.00 & 0.09 & 0.15 & 0.01 & 0.16 & 0.30 & -1.25 & 2.03 & 2.08\\ 
& &  {0.00} ({0.01}) & {0.08} ({0.01}) & {0.14} ({0.02}) & {0.01} ({0.01}) & {0.16} ({0.04}) & {0.29} ({0.07}) & {-1.16} ({0.30}) & {2.02} ({0.14}) & {2.15}  ({0.24}) \\ 
& &  {0.00} ({0.00}) & {0.08} ({0.01}) & {0.14} ({0.02}) & {0.00} ({0.00}) & {0.16} ({0.04}) & {0.27} ({0.06}) & {-1.35} ({0.31}) & {21.01} ({11.72}) & {33.20} ({19.26}) \\[.1cm] 
&1.2 & 0.00 & 0.09 & 0.15 & 0.01 & 0.16 & 0.29 &  & &  \\ 
& &  {0.00}  ({0.01}) &  {0.08}  ({0.01}) &  {0.14}  ({0.03}) &  {0.00}  ({0.01}) &  {0.16}  ({0.05}) &  {0.29}  ({0.08}) &  & &  \\ 
& &  {0.00}  ({0.00}) & {0.08}  ({0.01}) & {0.14}  ({0.02}) & {0.01}  ({0.00}) & {0.16}  ({0.04}) & {0.27}  ({0.07}) & & & \\[.1cm] 
		\end{tabular}
	}
	\label{tab:CLT-t9}
	\end{table}

\subsection{Simulations for nonlinear spectral statistics}
\label{subsec:emp:NLSS}

This section shows that the proposed bootstrap procedure can work for statistics beyond the class of linear spectral statistics. Recall that in order to apply Algorithm \ref{algorithm:bootstrap} to a generic nonlinear spectral statistic, say $\psi\{\lambda_1(\hat\Sigma_n),\ldots,\lambda_p(\hat\Sigma_n)\}$, it suffices to change only the third step to compute bootstrap samples of the form $\psi(\hat\lambda_1^*,\ldots,\hat\lambda_p^*)$. Even though a theoretical assessment for nonlinear spectral statistics is not feasible in the present paper, simulations have been conducted with the following examples: $T_{\max}=\lambda_1(\hat\Sigma_n)$, the largest sample eigenvalue, $T_{10}=\lambda_1(\hat\Sigma_n)+\cdots+\lambda_{10}(\hat\Sigma_n)$, the sum of the top ten sample eigenvalues, and $T_{\text{gap}}=\lambda_1(\hat\Sigma_n)-\lambda_2(\hat\Sigma_n)$, the spectral gap.

Note that asymptotic formulas for the distributions of these statistics are difficult to come by in many situations, especially if the matrix $Z$ is non-Gaussian, or if the matrix $\Sigma_n$ is non-diagonal.  The simulations were set up in essentially the same way as in the previous subsection for linear spectral statistics, except  that results are reported for  $\{T_{\max}-E(T_{\max})\}$, $\{T_{10}-E(T_{10})\}$, and $\{T_{\text{gap}}-E[T_{\text{gap}}]\}$. Note that unlike the simulations for linear spectral statistics, a factor of $p$ is omitted so that results are displayed on a convenient scale. Also, in this context, the bootstrap samples are centered by their empirical mean, rather than $\tilde{\vartheta}_n(f)$. The results are displayed in Tables \ref{tab:NLSS-Gaussian}--\ref{tab:NLSS-t9}, and they show an exciting picture: The proposed bootstrap algorithm worked well for each of the cases considered, which indicates the potential applicability of the proposed method beyond the class of linear spectral statistics. 
~\\
~\\
~\\
\begin{table}
\caption{Case (1): Standard Gaussian variables $Z_{ij}$.  Results are displayed as in Table \ref{tab:CLT-Gaussian} but with various nonlinear spectral statistics.}{
\scriptsize
\begin{tabular}{cc@{\quad}lll@{\quad}lll@{\quad}lll}
&& \multicolumn{3}{c}{\!\!\!\!\!\!\!\!\!\!\!\!\!\!\!\!\!\!\!\!\!\!\!\!$T_{\max}-E(T_{\max})\phantom{\displaystyle\frac 1n}$} & \multicolumn{3}{c}{\!\!\!\!\!\!\!\!\!\!\!\!\!\!\!\!\!\!\!\!\!\!\!\!$T_{10}-E(T_{10})$} & \multicolumn{3}{c}{\!\!\!\!\!\!\!\!\!\!\!\!\!\!\!\!\!\!\!\!\!\!\!\!$T_{\text{gap}}-E(T_{\text{gap}})$} \\[.3cm]
$\Sigma_n$ & $\gamma_n$ & sd & 95th & 99th & sd & 95th & 99th & sd & 95th & 99th \\[.1cm] 
(a) & 0.4 & 0.13 & 0.22 & 0.30 & 0.58 & 0.95 & 1.30 & 0.12 & 0.23 & 0.31 \\
&& {0.15} ({0.02}) & {0.25} ({0.05}) & {0.38} ({0.08}) & {0.57} ({0.05}) & {0.94} ({0.11}) & {1.34} ({0.20}) & {0.14} ({0.03}) & {0.25} ({0.05}) & {0.39} ({0.09}) \\[.1cm]
& 0.8 & 0.13 & 0.25 & 0.33 & 0.52 & 0.89 & 1.19 & 0.13 & 0.24 & 0.31 \\
&& {0.14} ({0.02}) & {0.25} ({0.05}) & {0.37} ({0.08}) & {0.53} ({0.05}) & {0.88} ({0.11}) & {1.26} ({0.19}) & {0.14} ({0.03}) & {0.25} ({0.06}) & {0.38} ({0.09}) \\[.1cm]
& 1.2 & 0.12 & 0.20 & 0.33 & 0.51 & 0.84 & 1.28 & 0.11 & 0.19 & 0.31 \\
&& {0.14} ({0.02}) & {0.24} ({0.04}) & {0.36} ({0.07}) & {0.49} ({0.04}) & {0.81} ({0.10}) & {1.18} ({0.19}) & {0.13} ({0.02}) & {0.24} ({0.05}) & {0.36} ({0.08}) \\[.1cm]
(b) & 0.4 & 0.06 & 0.11 & 0.15 & 0.10 & 0.17 & 0.26 & 0.07 & 0.12 & 0.17 \\
&& {0.06} ({0.01}) & {0.10} ({0.02}) & {0.15} ({0.03}) & {0.10} ({0.01}) & {0.17} ({0.02}) & {0.25} ({0.04}) & {0.07} ({0.01}) & {0.12} ({0.02}) & {0.17} ({0.03}) \\[.1cm]
& 0.8 & 0.06 & 0.10 & 0.14 & 0.10 & 0.17 & 0.23 & 0.07 & 0.12 & 0.17 \\
&& {0.06} ({0.01}) & {0.10} ({0.01}) & {0.15} ({0.03}) & {0.10} ({0.01}) & {0.17} ({0.02}) & {0.24} ({0.04}) & {0.07} ({0.01}) & {0.12} ({0.02}) & {0.18} ({0.03}) \\[.1cm]
& 1.2 & 0.06 & 0.11 & 0.14 & 0.10 & 0.17 & 0.25 & 0.08 & 0.13 & 0.20 \\
&& {0.06} ({0.01}) & {0.10} ({0.02}) & {0.15} ({0.03}) & {0.10} ({0.01}) & {0.17} ({0.02}) & {0.25} ({0.04}) & {0.07} ({0.01}) & {0.12} ({0.02}) & {0.17} ({0.03}) \\[.1cm]
(c) & 0.4 & 0.06 & 0.10 & 0.16 & 0.07 & 0.12 & 0.16 & 0.07 & 0.11 & 0.16 \\
&& {0.06} ({0.01}) & {0.11} ({0.01}) & {0.15} ({0.03}) & {0.07} ({0.01}) & {0.12} ({0.02}) & {0.17} ({0.03}) & {0.07} ({0.01}) & {0.11} ({0.01}) & {0.16} ({0.03}) \\[.1cm]
& 0.8 & 0.06 & 0.10 & 0.18 & 0.07 & 0.11 & 0.18 & 0.07 & 0.11 & 0.17 \\
&& {0.06} ({0.01}) & {0.11} ({0.01}) & {0.15} ({0.03}) & {0.07} ({0.01}) & {0.11} ({0.01}) & {0.16} ({0.03}) & {0.07} ({0.01}) & {0.11} ({0.01}) & {0.16} ({0.03}) \\[.1cm]
& 1.2 & 0.06 & 0.11 & 0.14 & 0.07 & 0.11 & 0.15 & 0.06 & 0.11 & 0.15 \\
&& {0.06} ({0.01}) & {0.11} ({0.02}) & {0.15} ({0.02}) & {0.07} ({0.01}) & {0.11} ({0.02}) & {0.16} ({0.03}) & {0.07} ({0.01}) & {0.11} ({0.02}) & {0.16} ({0.03}) \\[.1cm]
\end{tabular}
}
\caption{
 For each value of $\gamma_n$, the first row corresponds to the population quantities, and the second row corresponds to the mean and standard deviation (in parentheses) for the bootstrap estimates.
}
\label{tab:NLSS-Gaussian}
\end{table}

\begin{table}
\caption{Case (2): Standardized beta(6,6) variables $Z_{ij}$. Results are displayed as in Table  \ref{tab:NLSS-Gaussian}}{
\scriptsize
\begin{tabular}{cc@{\quad}lll@{\quad}lll@{\quad}lll}
&& \multicolumn{3}{c}{\!\!\!\!\!\!\!\!\!\!\!\!\!\!\!\!\!\!\!\!\!\!\!\!$T_{\max}-E(T_{\max})\phantom{\displaystyle\frac 1n}$} & \multicolumn{3}{c}{\!\!\!\!\!\!\!\!\!\!\!\!\!\!\!\!\!\!\!\!\!\!\!\!$T_{10}-E(T_{10})$} & \multicolumn{3}{c}{\!\!\!\!\!\!\!\!\!\!\!\!\!\!\!\!\!\!\!\!\!\!\!\!$T_{\text{gap}}-E(T_{\text{gap}})$} \\[.3cm]
$\Sigma_n$ & $\gamma_n$ & sd & 95th & 99th & sd & 95th & 99th & sd & 95th & 99th \\[.1cm] 
(a) & 0.4 & 0.13 & 0.23 & 0.32 & 0.55 & 0.95 & 1.29 & 0.11 & 0.22 & 0.35 \\
&& {0.14} ({0.02}) & {0.24} ({0.04}) & {0.35} ({0.06}) & {0.51} ({0.04}) & {0.84} ({0.10}) & {1.21} ({0.18}) & {0.13} ({0.02}) & {0.24} ({0.05}) & {0.37} ({0.08}) \\[.1cm]
& 0.8 & 0.13 & 0.21 & 0.31 & 0.53 & 0.89 & 1.22 & 0.12 & 0.20 & 0.36 \\
&& {0.14} ({0.02}) & {0.23} ({0.04}) & {0.35} ({0.07}) & {0.48} ({0.04}) & {0.80} ({0.10}) & {1.15} ({0.17}) & {0.13} ({0.03}) & {0.24} ({0.05}) & {0.37} ({0.09}) \\[.1cm]
& 1.2 & 0.12 & 0.21 & 0.33 & 0.50 & 0.79 & 1.23 & 0.12 & 0.22 & 0.34 \\
&& {0.12} ({0.02}) & {0.22} ({0.03}) & {0.32} ({0.06}) & {0.41} ({0.03}) & {0.68} ({0.08}) & {0.97} ({0.14}) & {0.12} ({0.02}) & {0.23} ({0.04}) & {0.34} ({0.07}) \\[.1cm]
(b) & 0.4 & 0.06 & 0.11 & 0.15 & 0.10 & 0.16 & 0.22 & 0.08 & 0.13 & 0.17 \\
&& {0.06} ({0.01}) & {0.09} ({0.01}) & {0.13} ({0.02}) & {0.09} ({0.01}) & {0.16} ({0.02}) & {0.22} ({0.03}) & {0.07} ({0.01}) & {0.11} ({0.01}) & {0.16} ({0.03}) \\[.1cm]
& 0.8 & 0.06 & 0.10 & 0.15 & 0.10 & 0.17 & 0.23 & 0.08 & 0.13 & 0.18 \\
&& {0.06} ({0.01}) & {0.09} ({0.01}) & {0.13} ({0.02}) & {0.09} ({0.01}) & {0.15} ({0.02}) & {0.22} ({0.03}) & {0.07} ({0.01}) & {0.11} ({0.01}) & {0.16} ({0.03}) \\[.1cm]
& 1.2 & 0.06 & 0.11 & 0.14 & 0.10 & 0.17 & 0.23 & 0.07 & 0.13 & 0.17 \\
&& {0.06} ({0.01}) & {0.09} ({0.01}) & {0.13} ({0.02}) & {0.09} ({0.01}) & {0.16} ({0.02}) & {0.22} ({0.03}) & {0.07} ({0.01}) & {0.11} ({0.02}) & {0.16} ({0.03}) \\[.1cm]
(c) & 0.4 & 0.06 & 0.09 & 0.14 & 0.06 & 0.11 & 0.16 & 0.06 & 0.10 & 0.15 \\
&& {0.06} ({0.01}) & {0.09} ({0.01}) & {0.14} ({0.02}) & {0.06} ({0.01}) & {0.11} ({0.01}) & {0.15} ({0.02}) & {0.06} ({0.01}) & {0.10} ({0.01}) & {0.14} ({0.02}) \\[.1cm]
& 0.8 & 0.06 & 0.11 & 0.14 & 0.06 & 0.11 & 0.14 & 0.06 & 0.11 & 0.14 \\
&& {0.06} ({0.01}) & {0.09} ({0.01}) & {0.14} ({0.02}) & {0.06} ({0.01}) & {0.10} ({0.01}) & {0.14} ({0.02}) & {0.06} ({0.01}) & {0.10} ({0.01}) & {0.14} ({0.02}) \\[.1cm]
& 1.2 & 0.06 & 0.10 & 0.15 & 0.06 & 0.11 & 0.15 & 0.06 & 0.10 & 0.15 \\
&& {0.06} ({0.01}) & {0.09} ({0.01}) & {0.13} ({0.02}) & {0.06} ({0.01}) & {0.10} ({0.01}) & {0.14} ({0.02}) & {0.06} ({0.01}) & {0.10} ({0.01}) & {0.14} ({0.02}) \\[.1cm]
\end{tabular}
}
\label{tab:NLSS-beta}
\end{table}

\begin{table}[h]
\caption{Case (3): Standardized t9 variables $Z_{ij}$. Results are displayed as in Table \ref{tab:NLSS-Gaussian}.}{
\scriptsize
		\begin{tabular}{cc@{\quad}lll@{\quad}lll@{\quad}lll}
			&& \multicolumn{3}{c}{\!\!\!\!\!\!\!\!\!\!\!\!\!\!\!\!\!\!\!\!\!\!\!\!$T_{\max}-E(T_{\max})\phantom{\displaystyle\frac 1n}$} & \multicolumn{3}{c}{\!\!\!\!\!\!\!\!\!\!\!\!\!\!\!\!\!\!\!\!\!\!\!\!$T_{10}-E(T_{10})$} & \multicolumn{3}{c}{\!\!\!\!\!\!\!\!\!\!\!\!\!\!\!\!\!\!\!\!\!\!\!\!$T_{\text{gap}}-E(T_{\text{gap}})$} \\[.3cm]
			$\Sigma_n$ & $\gamma_n$ & sd & 95th & 99th & sd & 95th & 99th & sd & 95th & 99th \\[.1cm] 
			(a)& 0.4 & 0.13 &  1.53 &  1.67 & 0.57 &  6.71 &  7.07 & 0.12 &  0.44 &  0.55 \\ 
			&& {0.17} ({0.03}) & {1.48} ({0.15}) & {1.64} ({0.13}) & {0.71} ({0.04}) & {6.97} ({0.12}) & {7.48} ({0.16}) & {0.15} ({0.03}) & {0.42} ({0.18}) & {0.58} ({0.16}) \\[.1cm] 
			& 0.8 & 0.13 &  2.07 &  2.16 & 0.53 &  12.64 &  12.98 & 0.11 &  0.41 &  0.54  \\ 
			&& {0.17} ({0.02}) & {2.02} ({0.15}) & {2.18} ({0.13}) & {0.65} ({0.04}) & {13.07} ({0.27}) & {13.54} ({0.29}) & {0.15} ({0.03}) & {0.40} ({0.17}) & {0.56} ({0.16}) \\[.1cm] 
			& 0.4 & 0.12 &  2.65 &  2.76 & 0.49 &  18.68 &  19.02 & 0.11 &  0.41 &  0.53 \\ 
			&& {0.16} ({0.03}) & {2.56} ({0.15}) & {2.71} ({0.14}) & {0.60} ({0.04}) & {19.30} ({0.46}) & {19.75} ({0.46}) & {0.14} ({0.03}) & {0.39} ({0.18}) & {0.54} ({0.16}) \\[.1cm] 		
			(b)& 0.4 & 0.06 &  0.18 &  0.22 & 0.10 &  0.89 &  0.98 & 0.07 &  0.13 &  0.19 \\ 
			&& {0.08} ({0.01}) & {0.20} ({0.01}) & {0.27} ({0.02}) & {0.13} ({0.01}) & {0.95} ({0.04}) & {1.05} ({0.05}) & {0.09} ({0.01}) & {0.15} ({0.02}) & {0.22} ({0.03}) \\[.1cm] 
			& 0.8 & 0.06 &  0.21 &  0.25 & 0.11 &  1.16 &  1.24 & 0.07 &  0.12 &  0.17  \\ 
			&& {0.08} ({0.01}) & {0.23} ({0.01}) & {0.29} ({0.03}) & {0.13} ({0.01}) & {1.22} ({0.04}) & {1.32} ({0.05}) & {0.09} ({0.01}) & {0.15} ({0.02}) & {0.22} ({0.03}) \\[.1cm] 
			& 0.4 & 0.06 &  0.21 &  0.24 & 0.10 &  1.34 &  1.41 & 0.07 &  0.11 &  0.15 \\ 
			&& {0.08} ({0.01}) & {0.25} ({0.01}) & {0.31} ({0.02}) & {0.13} ({0.01}) & {1.42} ({0.04}) & {1.52} ({0.05}) & {0.09} ({0.01}) & {0.14} ({0.02}) & {0.21} ({0.03}) \\[.1cm] 
			(c)& 0.4 & 0.08 &  0.14 &  0.20 & 0.09 &  0.17 &  0.24 & 0.08 &  0.13 &  0.20 \\ 
			&& {0.08} ({0.00}) & {0.14} ({0.01}) & {0.20} ({0.02}) & {0.09} ({0.00}) & {0.17} ({0.01}) & {0.24} ({0.02}) & {0.08} ({0.01}) & {0.13} ({0.01}) & {0.20} ({0.02}) \\[.1cm] 
			& 0.8 & 0.08 &  0.14 &  0.21 & 0.09 &  0.16 &  0.23 & 0.08 &  0.15 &  0.21  \\ 
			&& {0.08} ({0.00}) & {0.14} ({0.01}) & {0.20} ({0.02}) & {0.08} ({0.00}) & {0.16} ({0.01}) & {0.23} ({0.02}) & {0.08} ({0.00}) & {0.14} ({0.01}) & {0.20} ({0.02}) \\[.1cm] 
			& 0.4 & 0.09 &  0.16 &  0.24 & 0.09 &  0.18 &  0.27 & 0.09 &  0.16 &  0.24 \\ 
			&& {0.08} ({0.00}) & {0.14} ({0.01}) & {0.20} ({0.03}) & {0.08} ({0.00}) & {0.16} ({0.01}) & {0.23} ({0.03}) & {0.08} ({0.00}) & {0.14} ({0.01}) & {0.21} ({0.03}) \\[.1cm] 
		\end{tabular}
	}
	\label{tab:NLSS-t9}
\end{table}

\subsection{Application to hypothesis testing}
\label{subsec:emp:hypo}

The bootstrap procedure may be applied to compute critical values for sphericity tests of the null hypothesis $H_0\colon\Sigma_n=I_p$.
Three types of test statistics were considered in the simulations: Likelihood ratio test statistic, $\mathrm{tr}(\hat\Sigma_n)-\log\det\hat\Sigma_n-p$, John's test statistic, $\mathrm{tr}[\{p\hat\Sigma_n/\mathrm{tr}(\hat\Sigma_n)-I_p\}^2]$, and Condition number, $\lambda_1(\hat\Sigma_n)/\lambda_p(\hat\Sigma_n)$.
The latter two examples are nonlinear spectral statistics, whereas the likelihood ratio test is a linear spectral statistic based on $f(x)=x-\log(x)-1$. Further background may be found in~\cite{Muirhead:2005} and~\citet[Sec.~10.8]{Anderson:2003}. 
 
To explain how the critical values are obtained via the bootstrap, suppose that the observations $X_1,\dots,X_n$ are generated under $H_0$. In this case, the eigenvalues of $\Sigma_n$ are known, with \smash{$\lambda_1(\Sigma)=\cdots=\lambda_p(\Sigma)=1$.} Consequently, the matrix $\tilde\Lambda_n$  in Algorithm 3.1 may be replaced with the identity matrix $I_p$. Meanwhile, the kurtosis estimate $\hat\kappa_n$ is still computed with the proposed formula~\eqref{EQ:EST-KURT}. In summary, for any test statistic of the form $T=\psi \{ \lambda_1(\hat\Sigma_n),\dots,\lambda_p(\hat\Sigma_n) \}$ for some generic function $\psi$, the $1-\alpha$ quantile may be estimated as follows:

\begin{algo}[Bootstrap critical values]
\begin{tabbing} \\
  \enspace For: $b=1$ to $b=B$ \\
   \qquad Generate a random matrix $Z^*\in\mathbb{R}^{n\times p}$ with i.i.d.~entries drawn from Pearson$(0,1,0,\hat{\kappa}_n)$. \\
   \qquad Compute the eigenvalues of the matrix $\hat{\Sigma}_n^*=\frac 1n (Z^*)\ttop (Z^*)$, and denote them by $\lambda_1^*,\dots,\lambda_p^*$.\\
   \qquad Compute the statistic $T_{n,b}^*=\psi(\lambda_1^*,\dots,\lambda_p^*)$. \\
   \enspace Output the empirical $1-\alpha$ quantile of the values $T_{n,1}^*,\ldots,T_{n,B}^*$.
\end{tabbing} 
\end{algo}

The procedure above was applied to data drawn from the three distributions (1)--(3), with $\alpha=0.05$ or $\alpha=0.01$. The simulations were organized analogously to those in Section~\ref{subsec:emp:LSS}. Table~\ref{tab:hypo:type1} below shows that the bootstrap leads to type-I error rates very close to the nominal ones.

To examine the power of the tests when bootstrap critical values are used, another set of simulations were carried out under the spiked alternative $\Sigma=\text{diag}\{ \lambda_1(\Sigma_n),\dots,\lambda_{10}(\Sigma_n),1,\dots,1\}$ with \smash{$\lambda_j(\Sigma_n)= \phi>1$} for each $j=1,\dots,10$. The value $\phi$ was chosen separately for each setting and test, so that the test achieved a power of either $90\%$ or $80\%$ with the true $5\%$ critical value. Table~\ref{tab:hypo:power} below shows that the bootstrap critical values and the true critical values led to nearly the same power.

\begin{table}
\begin{center}
\caption{Type I errors of the bootstrap procedure at the 5\% and 1\% nominal levels for various tests of sphericity under various distributions and aspect ratios.}{
\begin{tabular}{lc@{\qquad}ll@{\qquad}ll@{\qquad}ll}
&& \multicolumn{2}{l}{\!\!\!LRT$\phantom{\displaystyle\frac 1n}$} & \multicolumn{2}{l}{\!\!\!John} & \multicolumn{2}{l}{\!\!\!CN} \\[.1cm]
$Z$ & $\gamma_n$ & 0.05 & 0.01 & 0.05 & 0.01 & 0.05 & 0.01 \\[.1cm]
Gaussian & 0.4 & 0.0452 & 0.0092 & 0.0471 & 0.0090 & 0.0477 & 0.0094 \\
& 0.8 & 0.0460 & 0.0086 & 0.0479 & 0.0119 & 0.0536 & 0.0117 \\
& 1.2 & & & 0.0530 & 0.0108\\[.1cm]
beta(6,6) & 0.4 & 0.0544 & 0.0133 & 0.0529 & 0.0137 & 0.0506 & 0.0092\\
& 0.8 & 0.0519 & 0.0090 & 0.0516 & 0.0116 & 0.0547 & 0.0107 \\
& 1.2 & & & 0.0512 & 0.0085 \\[.1cm]
t9 & 0.4 & 0.0497 & 0.0104 & 0.0524 & 0.0090 & 0.0505 & 0.0093 \\ 
& 0.8 & 0.0499 & 0.0103 & 0.0505 & 0.0106 & 0.0514 & 0.0109 \\ 
& 1.2 &  &  & 0.0516 & 0.0106 & &   \\[.1cm] 
\end{tabular}
}
\caption{
LRT and CN stands for likelihood ratio test and condition number, respectively. Note that Gaussian, beta(6,6), and t9 refer to the standardized versions of these distributions, with mean 0 and variance 1.
}
\label{tab:hypo:type1}
\end{center}
\end{table}

\begin{table}
	\begin{center}
	\caption{Power results}{
	\begin{tabular}{lc@{\qquad}ll@{\qquad}ll@{\qquad}ll}
		&& \multicolumn{2}{l}{\!\!\!LRT$\phantom{\displaystyle\frac 1n}$} & \multicolumn{2}{l}{\!\!\!John} & \multicolumn{2}{l}{\!\!\!CN} \\[.1cm]
		$Z$ & $\gamma_n$ & 0.90 & 0.80 & 0.90 & 0.80 & 0.90 & 0.80 \\[.1cm]
		Gaussian &0.4& 0.8972 & 0.7931 & 0.9074 & 0.7920 & 0.9088 & 0.8167 \\  
		& 0.8& 0.8880 & 0.7750 & 0.8859 & 0.7739 & 0.9112 & 0.8232 \\  
		& 1.2&  & & 0.8916 & 0.7915 &  &  \\[.1cm]
		beta(6,6) &0.4& 0.8764 & 0.7713 & 0.8988 & 0.7848 & 0.8902 & 0.7887 \\  
		& 0.8& 0.8982 & 0.8046 & 0.8946 & 0.7917 & 0.9065 & 0.8187 \\  
		& 1.2&  & & 0.8882 & 0.7859 &  &  \\[.1cm]
t9 & 0.4 & 0.8829 & 0.7873 & 0.8920 & 0.7971 & 0.9191 & 0.8119 \\ 
& 0.8 & 0.8959 & 0.8039 & 0.8877 & 0.7980 & 0.8677 & 0.7233 \\ 
& 1.2 &  &  & 0.8927 & 0.7919 & &   \\[.1cm] 
	\end{tabular}
	}
	\label{tab:hypo:power}
	\end{center}
\end{table}

\subsection{Protein data example}
\label{subsec:realdata}

The tumor suppressor protein p53 plays a fundamental role in human cancer research. Due to the fact that thousands of mutations of this protein have been observed in cancer patients, it is of interest to know how the properties of the protein vary across mutations. To address this general question, the paper~\cite{Danziger:2006} proposed a method to assign biophysical features to a large collection of p53 mutations. In particular, the authors used the method to produce a dataset of 31,159 mutations, with 5,408 features per mutation --- which is accessible as the `p53 mutants' dataset in the~\cite{UCI} repository. As an illustration, the following experiments consider the problem of detecting correlations among these features in a variety of scenarios.

The 31,159 rows of the p53 dataset were viewed as a finite population, and new datasets of varying sizes were obtained by sampling from this population.  More specifically, for each of the pairs $(n,p)\in\{150,250,500\}\times \{25,75,125\}$, a dataset of size $n\times p$ was obtained by sampling $n$ rows without replacement. In order to make the problem of detecting correlations more challenging, the $p$ columns corresponded to the variables $j\in\{1,\dots,5,\!408\}$ with the $p$ smallest correlation scores $\rho(j)$, defined by $\rho(j)=\sum_{l=1}^{5,408} |R_{jl}|$, where $R\in\R^{5,408\times 5,408}$ denotes the sample correlation matrix of the full p53 dataset. After each  $n\times p$ matrix was drawn, it was then standardized, and each of the three sphericity tests were applied to compute p-values. The calculations were done in the same manner as in Section~\ref{subsec:emp:hypo}, except that a large choice of $B=10^4$ bootstrap replicates was used in order to resolve very small p-values. Finally, in order to illustrate the typical performance of the tests, this entire process was repeated 500 times for each pair $(n,p)$, and then the 500 p-values from each test statistic were respectively averaged. The results are displayed in Table~\ref{tab:protein} below.

\begin{table}[h]
\begin{center}
\caption{Averaged p-values obtained from the sphericity tests (John, LRT, CN) over $500$ repeated experiments.}{
\begin{tabular}{c|ccc} 
	$(n,p)$ & John & LRT & CN \\  
	(500,125) & 0.001211 & 0.000009 & 0.000000 \\ 
	(500,100) & 0.001668 & 0.000020 & 0.000000 \\ 
	(500,75) & 0.002405 & 0.000044 & 0.000000 \\ 
	(250,125) & 0.002023 & 0.000009 & 0.000000 \\ 
	(250,100) & 0.002835 & 0.000020 & 0.000000 \\ 
	(250,75) & 0.004005 & 0.000055 & 0.000000 \\ 
	(150,125) & 0.002660 & 0.000012 & 0.000000 \\ 
	(150,100) & 0.003543 & 0.000018 & 0.000000 \\ 
	(150,75) & 0.005271 & 0.000054 & 0.000000 \\ 
\end{tabular}
}
\label{tab:protein}
\end{center}
\end{table}
Some interesting patterns are apparent in Table~\ref{tab:protein}. For a fixed $n$, the p-values become monotonically larger as $p$ decreases, which is intuitive because there are fewer possible correlations to detect. 
Another pattern is that for every pair $(n,p)$, all three tests obey the same ordering of power --- with the CN test being most powerful, and the John test being least powerful. Values of $p$ larger than 125 were also considered, but the cases of $p=75,100,$ and $125$ were selected for presentation, since they reveal prominent differences among the tests.

\section{Discussion}
\label{sec:conclusions}
In this paper, the Spectral Bootstrap procedure was proposed for approximating the distributions of spectral statistics in the high-dimensional setting. 
While the method is conceptually based on ideas from random matrix theory, the method is user-friendly since its implementation requires no knowledge of this subject matter. The main theoretical contribution states the consistency of the Spectral Bootstrap for linear spectral statistics. Simulation studies with a number of linear spectral statistics indicate that the method has excellent finite sample behavior for a range of distributions and varying kurtosis. Moreover, the method has the promise of being applicable beyond the class of linear spectral statistics, as evidenced through experiments with several nonlinear spectral statistics. This may be particularly useful in applications where formulas for limit laws do not yet exist. Future research may look into theoretically and computationally extending the scope of the proposed bootstrap. 

\section*{Acknowledgements} Lopes was partially supported by NSF grant DMS 1613218.  Aue was partially supported by NSF grants DMS 1305858 and DMS 1407530. We thank Debashis Paul for helpful feedback.


\bibliographystyle{rss}
\bibliography{LSS_boot_bib.bib}

\appendix
~\\

\section*{Appendices} The proofs are organized according to the order of the results in the main text.
~\\

\noindent \emph{Notation.} Before presenting the proofs, we first mention a few notational items. If $a_n$ and $b_n$ are numerical sequences, we write $a_n\lesssim b_n$, or equivalently $a_n=\mathcal{O}(b_n)$,  if there is a positive constant $c$ such that $|a_n|\leq c |b_n|$ for all large $n$ (where $a_n$ and $b_n$ are allowed to be complex). Likewise, the expression $a_n\asymp b_n$ means $a_n\lesssim b_n$ and $b_n\lesssim a_n$. Lastly, define the upper and lower complex half-planes $\C^+=\{z\in\C\colon \Im(z)>0\}$ and $\C^-=\{z\in\C\colon \Im(z)<0\}$.

\section{Proof of Theorem \ref{THM:CONSIST}}
\label{sec:proofs:kappa}
Below, we prove the consistency of $\hat\kappa_n$ in Section~\ref{sec:kappaconsist}, and the consistency of $\tilde H_n$ in Section~\ref{sec:Hconsist}.
\subsection{Consistency of kurtosis estimator}\label{sec:kappaconsist}

Recall that a generic estimator $\hat\theta_n$ for a parameter $\theta_n$ is said to be ratio-consistent if $\hat\theta_n/\theta_n\xrightarrow{ \ \P \ } 1$. Lemmas~\ref{lem:omega_consist},~\ref{lem:nu_consist}, and~\ref{LEM:TAU_CONSIST} will establish the ratio-consistency of $\hat\omega_n$, $\hat\nu_n$, and $\hat\tau_n$ respectively. In proving these lemmas, we will rely on some facts about random quadratic forms, which are summarized in the following lemma obtained from~\citet[Lemma B.26]{Bai:Silverstein:2010} and~\citet[eqn.\,1.15]{Bai:Silverstein:2004}.

\begin{lemma}\label{lemma:qforms}
Let $A\in \R^{p\times p}$ be a non-random matrix, and let $V\in\R^p$ be  a random vector with independent entries satisfying $E(V_j)=0$, $E(V_j^2)=1$, $E(V_j^4)=\kappa$. Also, let $r\in [1,\infty)$ be fixed, and suppose $E(|V_j|^s)\leq c_s$ for $1\leq s\leq 2r$. Then,
\begin{equation}
E\Big( \big|V\ttop A V -\text{\emph{tr}}(A)\big|^r\Big) \leq C_r \Big[\kappa^{r/2}\,\|A\|_F^r+c_{2r}\tr\big\{(AA\ttop)^{r/2}\big\}\Big],
\end{equation}
where $C_r>0$ is a number depending only on $r$. Furthermore, in the case $r=2$, the following formula holds:
\begin{equation}\label{EQ:KURTOSISFORMULA}
\text{\emph{var}}(V\ttop AV) = 2\|A\|_F^2+(\kappa-3)\ts\sum_{j=1}^p A_{jj}^2.
\end{equation}
 
\end{lemma}

\begin{lemma}\label{lem:omega_consist}
Suppose Assumption~\ref{assumption:data} holds. Then, as $n\to\infty$,
\begin{equation}
\ts\frac{1}{\omega_n}E(|\hat\omega_n-\omega_n|)\to 0.
\end{equation}
\end{lemma}
\begin{proof} For each $j=1,\dots,p$, define the estimator
\[
\ts\hat\sigma_j^2= n^{-1} \sum_{i=1}^nX_{ij}^2,
\]
which clearly satisfies $E(\hat\sigma_j^2)=\sigma_j^2$, and allows $\hat\omega_n$ to be written as
$$\hat\omega_n= \sum_{j=1}^p (\hat\sigma_j^2)^2.$$
Considering the bound
\begin{equation}\label{EQ:TRIANGLEBOUND}
E\big(|\hat\omega_n-\omega_n|\big) \leq \sum_{j=1}^p E\big\{|(\hat\sigma_j^2)^2 -\sigma_j^4|\big\},
\end{equation}
we concentrate on the $j$th term,
\begin{equation}\label{EQ:TERMWISEBOUND}
 \begin{split}
 E\Big\{\big|(\hat\sigma_j^2)^2 -\sigma_j^4\big| \Big\}&=E\Big\{|\hat\sigma_j^2-\sigma_j^2|\cdot(\hat\sigma_j^2+\sigma_j^2)\Big\}\\[0.2cm]
  &\leq \sqrt{\var(\hat\sigma_j^2)}\cdot \sqrt{E\Big[\big\{(\hat\sigma_j^2-\sigma_j^2)+2\sigma_j^2)\big\}^2\Big]}\\[0.2cm]
 &\leq \sqrt{\var(\hat\sigma_j^2)}\cdot \sqrt{2 \var(\hat\sigma_j^2)+2 (2\sigma_j^2)^2},
 \end{split}
\end{equation}
where in the last step we have used the general inequality $E\{(U+V)^2\}\leq 2E(U^2)+2E(V^2)$. 
We now deal with the problem of bounding $\var(\hat\sigma_j^2)$. Let $v_j=\Sigma^{1/2}e_j\in\mathbb{R}^p$, and define the rank-1 matrix $A^{(j)}=v_jv_j\ttop$. In turn, letting $Z_{i\cdot}$ denote the $i$th row of $Z$, we have
\[
X_{ij}^2
=\big(e_i^\top Z\Sigma^{1/2}e_j\big)^2
=\big(Z_{i\cdot}^\top v_j\big)^2
=Z_{i\cdot}^\top A^{(j)} Z_{i\cdot},
\]
and since the matrix $Z$ has i.i.d.~entries, it follows from Lemma~\ref{lemma:qforms} that
\begin{align*}
\mathrm{var}(\hat\sigma_j^2)
&=\ts n^{-1} \mathrm{var}(Z_{1\cdot}^\top A^{(j)}Z_{1\cdot})  \    \lesssim \ \ts n^{-1} \|A^{(j)}\|_F^2 \ 
= \ \ts n^{-1} \|v_j\|_2^4  \ =\ts n^{-1}\sigma_j^4.
\end{align*}
Now, returning to the bounds~\eqref{EQ:TRIANGLEBOUND} and~\eqref{EQ:TERMWISEBOUND}, we see that
\begin{equation*}
\begin{split}
\ts\frac{1}{\omega_n}E\big(|\hat\omega_n-\omega_n|\big) \ & \lesssim \  \ts \frac{1}{\omega_n}\displaystyle  \sum_{j=1}^p \ts\frac{1}{\sqrt{n}} \sigma_j^2 \cdot \sqrt{\ts\frac{2}{n}\sigma_j^4+8\sigma_j^4}\\[0.2cm]
& \lesssim  \ts \frac{1}{\omega_n}\displaystyle  \sum_{j=1}^p \ts\frac{1}{\sqrt{n}} \sigma_j^4\\[0.2cm]
&= \ts\frac{1}{\sqrt{n}},
\end{split}
\end{equation*}
which completes the proof.
\end{proof}

\begin{lemma}\label{lem:nu_consist}
Suppose Assumption~\ref{assumption:data} holds. Then, $E(\hat\nu_n)=\nu_n$, and as $n\to\infty$,
\[
\mathrm{var}\bigg(\frac{\hat\nu_n}{\nu_n}\bigg)\xrightarrow{ \ \ } 0.
\]
\end{lemma}

\begin{proof} The unbiasedness of $\hat{\nu}_n$ is clear. It is a classical fact that for a generic i.i.d.~sample $Y_1,\ldots,Y_n$ of scalar variables, the sample variance $\hat{\varsigma}^2=\frac{1}{n-1}\sum_{i=1}^n (Y_i-\bar{Y})^2$ satisfies
\begin{equation}\label{EQ:VARVAR}
\mathrm{var}\Big(\frac{\hat{\varsigma}^{\, 2}}{\varsigma^2}\Big) \, \lesssim \, \frac{1}{n}\frac{\mu_4}{\varsigma^4}
\end{equation}
where $\mu_4$ is the fourth central moment of $Y_1$, and $\varsigma^2=\mathrm{var}(Y_1)$~\citep[p.\,164]{Kenney:1951}. If we let $Y_i=\|X_{i\cdot}\|_2^2$, where $X_{i\cdot}$ denotes the $i$th row of $X$,  then we have $\varsigma^2=\nu_n$. Using the formula~\eqref{EQ:VARVAR}, it remains to show that
\begin{equation} \label{FOURTHPOWERLIMIT}
\frac{1}{n \nu_n^2} E\Big[\big \{\|X_{1\cdot}\|_2^2-\mathrm{tr}(\Sigma_n)\big\}^4\Big] \xrightarrow{ \ } 0.
\end{equation}
Noting that $\|X_{1\cdot}\|_2^2=Z_{1\cdot}\ttop \Sigma Z_{1\cdot}$, and that $\tr(\Sigma_n^4)\leq \|\Sigma_n\|_F^4$, we may apply Lemma~\ref{lemma:qforms} to conclude that
\begin{equation*}
E\Big[\big\{\|X_{1\cdot}\|_2^2-\mathrm{tr}(\Sigma_n)\big\}^4\Big]\lesssim \|\Sigma_n\|_F^4.
\end{equation*}
Furthermore,  since we assume $\kappa>1$, the formula~\eqref{EQ:KURTOSISFORMULA} in Lemma~\ref{lemma:qforms} implies that $\nu_n \gtrsim \|\Sigma_n\|_F^2$.
Hence, the limit~\eqref{FOURTHPOWERLIMIT} holds with rate $\mathcal{O}(1/n)$.
\end{proof}

\begin{lemma}\label{LEM:TAU_CONSIST}
Suppose that Assumptions~\ref{assumption:data}and~\ref{assumption:esd} hold. Then, as $n\to\infty$
\begin{equation}
\hat\tau_n/\tau_n\xrightarrow{ \ \P \ } 1. 
\end{equation}
\end{lemma}

\begin{proof}
We refer to~\citet[Section A.3]{Bai:Saranadasa:1996} for the proof.
\end{proof}

We now assemble the previous lemmas to show that $\hat \kappa_n\xrightarrow{\P}\kappa$.
 As a preliminary step, we check that each of the quantities $\nu_n$, $\tau_n$, and $\omega_n$ are of the same order. In the case of $\tau_n$, we have $\tau_n\asymp p$, since $\tau_n=\sum_{j=1}^p\lambda_j^2(\Sigma_n)$, and each eigenvalue is bounded away from 0 and $\infty$ by Assumption~\ref{assumption:esd}. For the same reason, we have $\omega_n=\sum_{j=1}^p \sigma_j^4\asymp p$, since 
$$\lambda_p(\Sigma_n)\ \leq  \ \min_{1\leq j\leq p}\sigma_j^2 \ \leq \ \max_{1\leq j\leq p}\sigma_j^2 \ \leq \  \lambda_1(\Sigma_n).$$
 Lastly, to check $\nu_n\asymp p$, recall the identity
\begin{equation}\label{EQN:TEMPIDENTITY}
\kappa=3+\ts\frac{\nu_n}{\omega_n}-\frac{2\tau_n}{\omega_n}.
\end{equation}
Since $\kappa$ is fixed and $\tau_n/\omega_n\asymp 1$, we have $\nu_n/\omega_n=\mathcal{O}(1)$. On the other hand, we can also see that $\nu_n/\omega_n$ is bounded below by a positive constant, due to the assumption that $\kappa>1$, and the fact that $\tau_n\geq \omega_n$. Thus $\nu_n/\omega_n\asymp 1$, and $\nu_n\asymp p$.

To proceed, define the quantity
\begin{equation*}
\breve{\kappa}_n:=3+\ts\frac{\hat\nu_n-2\hat\tau_n}{\hat\omega_n}.
\end{equation*}
Since the function  $\max\{\cdot,1\}$ is continuous, the proof may be completed by showing that $\breve \kappa_n\xrightarrow{\P}\kappa$. Using the fact the parameter estimates $\hat \tau_n$, $\hat\nu_n$, and $\hat\omega_n$ are individually ratio-consistent, it follows that if we fix $\epsilon\in (0,1)$, then the following event has probability tending to 1,
\begin{equation}
\begin{split}
\breve\kappa_n &\leq 3+ \ts\frac{(1+\epsilon)\nu_n}{(1-\epsilon)\omega_n}-\ts\frac{(1-\epsilon)2\tau_n}{(1+\epsilon)\omega_n}.
\end{split}
\end{equation}
Consequently, the identity~\eqref{EQN:TEMPIDENTITY} implies there is an absolute constant $C>0$ such that the event
\begin{equation}
\breve \kappa_n \leq \kappa+ C \cdot\epsilon\cdot\Big(\ts\frac{\nu_n+\tau_n}{\omega_n}\Big)
\end{equation}
has probability tending to 1. Moreover, since our earlier work ensures $\ts\frac{\nu_n+\tau_n}{\omega_n}\asymp 1$, there is a possibly larger absolute constant $C>0$, such that the event
$\{\breve\kappa_n\leq \kappa+C\epsilon\}$
has probability tending to 1. Finally, a symmetric argument shows that the event $\{\breve\kappa_n \geq \kappa-C\epsilon\}$ also has probability tending to 1, which completes the proof.\qed

\subsection{Consistency of spectrum estimator}\label{sec:Hconsist}

Define the empirical distribution function associated with the QuEST eigenvalues,
\begin{equation*}
\hat H_{\text{Q},n}(\lambda)=\ts\frac 1p \sum_{j=1}^p 1\{\hat\lambda_{\text{Q},j}\leq \lambda\}.
\end{equation*}
Under our assumptions, the proof of Theorem 2.2 in~\cite{Ledoit:Wolf:2015} shows that the following limit holds almost surely
\begin{equation}\label{EQ:QUESTCONSIST}
 \hat H_{\text{Q},n} \Rightarrow H.
\end{equation}
To prove the almost-sure limit $\tilde H_n\Rightarrow H$,  let the random variable $N_n$ denote the number of values $\tilde\lambda_j$ that differ from their QuEST counterpart $\hat\lambda_{\text{Q},j}$. In this notation, it is sufficient to show that $N_n=o(p)$ almost surely, because this implies that for any fixed $\lambda$, the following relation holds almost surely,
 $$\tilde H_n(\lambda) = \hat H_{\text{Q},n}(\lambda)+o(1),$$
and it then follows from a short argument that $\tilde H_n\Rightarrow H$ almost surely.
 
 To show that $N_n=o(p)$ almost surely, first note that $N_n$ can be written as
 \begin{equation}\label{EQ:NDEF}
 N_n=\ts\sum_{j=1}^p1\big\{\hat\lambda_{\text{Q},j} > \hat \lambda_{\text{bound},n}\big\},\end{equation}
 where we recall $\hat\lambda_{\text{bound},n}=2\lambda_1(\hat\Sigma_n)$.
 Next, we claim it is sufficient to show that 
 \begin{equation}\label{EQ:TEMPCLAIM}
 \liminf_{n\to\infty}\big\{2\lambda_1(\hat\Sigma_n)-\lambda_1(\Sigma_n)\big\} \geq \epsilon
 \end{equation}
  holds almost surely, for some positive number $\epsilon$. To see why, consider the random variable
 $$N_n'=\sum_{j=1}^p1\{\hat\lambda_{\textup{Q},j}>\lambda_1(\Sigma_n)+\ts\frac{\epsilon}{2}\},$$
and note that~\eqref{EQ:TEMPCLAIM} implies the following asymptotic bound holds almost surely,
 $$N_n\leq N_n'+o(1).$$
 In turn, the condition~\eqref{EQ:QUESTCONSIST} and the assumption $H_n\Rightarrow H$ imply $N_n'=o(p)$ almost surely, which leads to the desired conclusion that $N_n=o(p)$ almost surely.

 To prove~\eqref{EQ:TEMPCLAIM}, let $u_1$ denote the top eigenvector of $\Sigma_n$. Then, we have the lower bound
 \begin{equation}\label{EQ:BOUNDING}
 \begin{split}
 \lambda_1(\hat\Sigma_n) &= \sup_{\|u\|_2=1} u\ttop \hat\Sigma_n u\\[0.2cm]
 %
 %
 %
 &\geq \lambda_1(\Sigma_n)\cdot  u_1\ttop\big(\ts\frac 1n Z\ttop Z\big)u_1,
 \end{split}
 \end{equation}
 which leads to
$$2\lambda_1(\hat\Sigma_n)-\lambda_1(\Sigma_n) \ \geq \ \Big\{2u_1\ttop\big(\ts\frac 1n Z\ttop Z\big)u_1 -1\Big\}\lambda_1(\Sigma_n).$$
Under our data-generating moel, it can be checked that $u_1\ttop\big(\ts n^{-1} Z\ttop Z\big)u_1\to 1$ almost surely. This can be done with the help of a strong law of large numbers for triangular arrays~\citep[Corollary 1]{Hu:1989}, as well as the moment bound in Lemma~\ref{lemma:qforms}. 
Meanwhile, due to Assumption~\ref{assumption:esd}, we know that $\liminf_{n\to\infty}\lambda_1(\Sigma_n)$ is bounded below by a positive constant, and so the last few steps imply~\eqref{EQ:TEMPCLAIM}.
 

 Finally, we prove that $\sup_n \lambda_1(\tilde\Lambda_n)<\infty$. Note that by the sub-multiplicative property of the operator norm,
 $$\hat\lambda_{\text{bound},n}\leq 2\cdot\lambda_1(\Sigma_n)\cdot \lambda_1(\ts\frac 1n Z\ttop Z).$$
Due to Assumption~\ref{assumption:esd}, we have $\sup_n\lambda_1(\Sigma_n)<\infty$. Also, under Assumption~\ref{assumption:data}, it is known from~\cite{Yin:Bai:1988} that $\sup_n\lambda_1(\ts\frac 1n Z\ttop Z)<\infty$ almost surely.\qed

\section{Proofs of Propositions~\ref{pro:spiked} and~\ref{pro:rankk}}

\begin{lemma}\label{LEM:STIELTJES}
For any $z\in \C\setminus\R$, any $j=1,\dots,p$, and any $\ell\in\{1,2\}$, the following bound holds,
$$\big|\{\Gamma_{n,\ell}(z)\}_{jj}\big| \ \leq \ \frac{\lambda_1(\Sigma_n)}{|\Im(z)|^{\ell}}.$$
\end{lemma}

\begin{proof} Since $|w|\geq |\Im(w)|$ for any complex number $w$, the definition of $\Gamma_n(z)$ implies
$$\big| \{ \Gamma_{n,\ell}(z) \}_{jj}\big| \ \leq \   \frac{\lambda_j(\Sigma)}{\big|\Im[z\big\{1+\lambda_j(\Sigma_n)\gamma_n t_n(z)\big\}\big]\big|^{\ell}}.$$
Next, we define the function
$$s_{n,j}(z)=\frac{-1}{z\{1+\lambda_j(\Sigma_n) \gamma_n t_n(z)\}}.$$
Because the function $t_n(z)$ is a Stieltjes transform, it is a fact that $s_{n,j}(z)$ is a Stieltjes transform of some distribution, as shown in the proof of Corollary 3.1 in the book~\cite{Couillet:Debbah:2011}. This implies that for any $z\in \C^+$,
$$\Im\{\ts\frac{1}{s_{n,j}(z)}\} \leq -\Im(z),$$ 
which is explained in
\citep[Theorem 3.2]{Couillet:Debbah:2011}. When the number $\Im(z)$ is positive, we have $|\Im\{\ts\frac{1}{s_{n,j}(z)}\}| \geq |\Im(z)|$, and hence
$$\big|\Im\big[z\{1+\lambda_j(\Sigma_n)\gamma_n t_n(z)\}\big]\big| \ \geq  \ |\Im(z)|,$$
yielding
$$\frac{\lambda_j(\Sigma_n)}{\big|\Im\big[z\{1+\lambda_j(\Sigma_n)\gamma_n t_n(z)\}\big]\big|^{\ell}} \  \leq \ \frac{\lambda_1(\Sigma_n)}{|\Im(z)|^{\ell}}.$$
The proof can be essentially repeated in the case when $\Im(z)$ is negative by using $\overline{s_{n,j}(z)}=s_{n,j}(\bar{z})$.
\end{proof}

\emph{Proof of Proposition~\ref{pro:spiked}.}
Let $z\in\C\setminus\R$ be fixed. Also, let the diagonal entries of $\Gamma_{n,\ell}(z)$ be written as $\{d_{1,\ell}(z),\dots,d_{p,\ell}(z)\}$, and let the columns of $U_n$ be denoted as $u_1,\dots,u_p$. Then, for any fixed $1\leq j\leq p$,
$$\{U_n\Gamma_{n,\ell}(z) U_n\ttop\}_{jj} = \sum_{l=1}^p d_{l,\ell}(z) (u_l\ttop e_j)^2.$$
Under a spiked covariance model, note that the entries $d_{k+1,\ell}(z),\dots, d_{p,\ell}(z)$ are all equal to $d_{p,\ell}(z)$. From the identity $\sum_{l=1}^p (u_l\ttop e_j)^2=1$, it follows that, for each $1\leq j\leq p$,
$$\{U_n\Gamma_{n,\ell}(z) U_n\ttop\}_{jj}= d_{p,\ell}(z) +a_{j,\ell}(z), \ \ \ \ \ \  \text{where }  \ \ \ \ \  a_{j,\ell}(z):=\ts\sum_{l=1}^k \{d_{l,\ell}(z)-d_{p,\ell}(z)\}(u_l\ttop e_j)^2.$$
By Lemma~\ref{LEM:STIELTJES}, we have the following bound for each $l=1,\dots,p$,
\begin{equation}\label{DIAGVALBOUND}
|d_{l,\ell}(z)|\leq \ts \frac{ \lambda_1(\Sigma_n)}{|\Im(z)|^{\ell}},
\end{equation}
 and so the numbers $a_{j,\ell}(z)$ satisfy 
\begin{equation}\label{EQ:KEYBOUND}
\ts\frac 1p \displaystyle \sum_{j=1}^p |a_{j,\ell}(z)| \ \leq \  \ts \frac{2\,\lambda_1(\Sigma)}{|\Im(z)|^{\ell}}\cdot \frac 1p\displaystyle \sum_{l=1}^k \sum_{j=1}^p (u_l\ttop e_j)^2 =\  \ts \frac{2\,\lambda_1(\Sigma)}{|\Im(z)|^{\ell}}\cdot \frac{k}{p}= o(1),
\end{equation}
and
\begin{equation}\label{EQ:KEYBOUND2}
 \max_{1\leq j\leq p} |a_{j,\ell}(z)| \leq  \ts \frac{2\,\lambda_1(\Sigma)}{|\Im(z)|^{\ell}}.
\end{equation}
Now, let $z_1,z_2\in\C\setminus\R$ be fixed, and consider the sum
\begin{equation*}
\begin{split}
\ts\frac 1p \displaystyle\sum_{j=1}^p \{U_n\Gamma_{n,\ell}(z_1) U_n\ttop\}_{jj}\{U_n\Gamma_{n,2}(z_2) U_n\ttop\}_{jj}&= \ts\frac 1p \displaystyle\sum_{j=1}^p \big\{d_{p,\ell}(z_1)+a_{j,\ell}(z_1)\big\}\big\{d_{p,2}(z_2)+a_{j,2}(z_2)\big\}\\[0.2cm]
&=d_{p,\ell}(z_1)d_{p,2}(z_2)+R_{n,\ell}(z_1,z_2),
\end{split}
\end{equation*}
where we define the remainder
$$R_{n,\ell}(z_1,z_2)=\ts \frac{d_{p,\ell}(z_1)}{p} \displaystyle \sum_{j=1}^p a_{j,2}(z_2) + \ts \frac{d_{p,2}(z_2)}{p}\displaystyle \sum_{j=1}^p a_{j,\ell}(z_1) +\ts\frac{1}{p}\displaystyle\sum_{j=1}^p a_{j,\ell}(z_1)a_{j,2}(z_2).$$
It follows that the bounds~\eqref{DIAGVALBOUND},~\eqref{EQ:KEYBOUND} and~\eqref{EQ:KEYBOUND2}, along with H\"older's inequality, imply 
$$R_{n,\ell}(z_1,z_2)=o(1).$$ 
Lastly, we must compare with the sum of the values $\{\Gamma_{n,\ell}(z_1)\}_{jj}\{\Gamma_{n,2}(z_2)\}_{jj}$. Observe that
\begin{equation*}
\begin{split}
\ts\frac 1p \displaystyle\sum_{j=1}^p \{\Gamma_{n,2}(z_1)\}_{jj}\{\Gamma_{n,\ell}(z_2)\}_{jj}&= \ts\frac 1p \displaystyle \sum_{j=1}^p d_{j,\ell}(z_1)d_{j,2}(z_2)\\
&= d_{p,\ell}(z_1)d_{p,2}(z_2)+\ts\frac 1p\displaystyle\sum_{j=1}^k \Big\{d_{j,\ell}(z_1)d_{j,2}(z_2)-d_{p,\ell}(z_1)d_{p,2}(z_2)\Big\}\\[0.2cm]
&=d_{p,\ell}(z_1)d_{p,2}(z_2)+\mathcal{O}(\ts\frac{k}{p}),
\end{split}
\end{equation*}
where we have again used the bound~\eqref{DIAGVALBOUND}. Altogether, this verifies desired limit.

%
%
%
%

%
%
\emph{Proof of Proposition~\ref{pro:rankk}.} Let $z\in\C\setminus\R$ be fixed. As before, let the diagonal entries of $\Gamma_{n,\ell}(z)$ be denoted as $\{d_{1,\ell}(z),\dots,d_{p,\ell}(z)\}$.  Observe that
\begin{equation}\label{RANKKEXPAND}
\begin{split}
\{U_n\Gamma_{n,\ell}(z)U_n\ttop\}_{jj} &= e_j\ttop(I_p-2\Pi)\Gamma_{n,\ell}(z)(I_p-2\Pi)e_j\\
&= e_j\ttop \Gamma_{n,\ell}(z)e_j-4e_j\ttop \Pi \, \Gamma_{n,\ell}(z)e_j+4e_j\ttop \Pi \, \Gamma_{n,\ell}(z)\Pi e_j.
\end{split}
\end{equation}
Since $\Gamma_{n,\ell}(z)e_j=d_{j,\ell}(z) e_j$, the middle term on the right side satisfies the bound
$$|e_j\ttop \Pi \, \Gamma_{n,\ell}(z)e_j| \ = \ |d_{j,\ell}(z)|\cdot |e_j\ttop \Pi e_j| \leq \ts\frac{\lambda_1(\Sigma_n)}{|\Im(z)|^{\ell}} \Pi_{jj},$$
where we have used Lemma~\ref{LEM:STIELTJES} in the second step.
Note that $\Pi_{jj}=e_j\ttop \Pi e_j$ is non-negative because $\Pi$ is necessarily positive semidefinite.
Similarly, Lemma~\ref{LEM:STIELTJES} also implies
$$|e_j\ttop \Pi\ttop \,\Gamma_{n,\ell}(z)\Pi e_j|  \leq \ts\frac{\lambda_1(\Sigma_n)}{|\Im(z)|^{\ell}} |e_j\ttop \Pi\ttop \Pi e_j|=\frac{\lambda_1(\Sigma_n)}{|\Im(z)|^{\ell}} \Pi_{jj}.$$
Hence, viewing $z$ as fixed, equation~\eqref{RANKKEXPAND} gives
$$\{U_n\Gamma_{n,\ell}(z)U_n\ttop\}_{jj} =d_{j,\ell}(z)+\mathcal{O}\big(\Pi_{jj}\big).$$
Likewise, using the boundedness of the values $d_j(z)$, and the fact that $(\Pi_{jj})^2 \leq  \Pi_{jj}$,  it follows that for any fixed numbers $z_1,z_2\in\C\setminus\R$,
$$\{U_n\Gamma_{n,\ell}(z_1)U_n\ttop\}_{jj} \{U_n\Gamma_{n,2}(z_2)U_n\ttop\}_{jj}  =  d_{j,\ell}(z_1)d_{j,2}(z_2)+\mathcal{O}\big(\Pi_{jj}\big).$$
Consequently, averaging over $j$ leads to 
\begin{equation*}
\begin{split}
\ts\frac{1}{p}\displaystyle\sum_{j=1}^p \{U_n\Gamma_{n,\ell}(z_1)U_n\ttop\}_{jj} \{U_n\Gamma_{n,2}(z_2)U_n\ttop\}_{jj}
&= \ts\frac{1}{p}\tr\{\Gamma_{n,\ell}(z_1)\Gamma_{n,2}(z_2)\}+\ts\frac{1}{p}\mathcal{O}\{\tr(\Pi)\},\\[0.2cm]
&=\ts\frac1p\displaystyle\sum_{j=1}^p \{\Gamma_{n,\ell}(z_1)\}_{jj} \{\Gamma_{n,2}(z_2)\}_{jj}+\ts\frac{1}{p}\mathcal{O}(\text{rank}(\Pi)),
\end{split}
\end{equation*}
which proves the desired limit.


\section{Proof of Theorem \ref{THM:BOOTSTRAP_CONSIST}}
\label{sec:proofs:bootstrap}

The whole proof in this section, as well as Section~D, is inspired by the arguments and results in~\cite{Najim:Yao:2016}. Likewise, familiarity with that paper is suggested for understanding the work here.

Let $\mathscr{C}_c^3(\R)$ denote the set of $3$-times continuously differentiable functions on $\R$ with compact support. For any function $f\in\mathscr{C}^3(\mathcal{I})$, there is another function $g\in\mathscr{C}_c^3(\R)$ such that $f=g$ on some open interval $\mathcal{I}'$ that satisfies $[a,b]\subset \mathcal{I}'\subset \mathcal{I}$.
Furthermore, due to the comments on page~\pageref{intervalcomments}, it is known that with probability 1, every eigenvalue of $\hat\Sigma_n$ lies in $\mathcal{I}'$ for all large $n$. It follows that $f$ and $g$ will asymptotically agree on all eigenvalues of $\hat{\Sigma}_n$. Hence, we may prove Theorem~\ref{THM:BOOTSTRAP_CONSIST} with the set $\mathscr{C}_c^3(\R)$ in place of  $\mathscr{C}^3(\mathcal{I})$.

For $z\in\mathbb{C}\setminus\R$, define the Stieltjes transforms
\begin{align*}
m_n(z) &=  \int \frac{1}{\lambda-z}dH_n(\lambda),\\[0.2cm]
\hat{m}_n(z) &= \ts\frac{1}{p}\mathrm{tr}\big\{(\hat{\Sigma}_n-zI_p)^{-1}\big\},\\[0.2cm]
\hat{m}_n^*(z;X) &=\ts\frac{1}{p}\mathrm{tr}\big\{(\hat{\Sigma}_n^* -zI_p)^{-1}\big\},
\end{align*}
where the matrix $X$ should be viewed as fixed when interpreting $\hat{m}_n^*(z;X)$ as the empirical Stieltjes transform of $\hat{\Sigma}_n^*$.
For any positive integer $r$,  define the operator $\Phi_r$ to act on a function $f\in \mathscr{C}_c^{r+1}(\R)$ according to
\begin{equation*}
\Phi_r(f)(z) =\sum_{l=0}^r \frac{(\sqrt{-1}y)^l}{l!}f^{(l)}(x)\chi(y),
\end{equation*}
where $z=x+\sqrt{-1}y$ and the function $\chi\colon\R\to\R^+$ is a particular cut-off function that is  smooth, compactly supported, and equal to 1 in a neighborhood of 0. In turn, for any fixed $f\in\mathscr{C}_c^{r+1}(\R)$, we formally define the linear functional $\phi_{f,r}$ to act on a test function $h\colon \C^+\to\C$ according to
$$\phi_{f,r}(h) =\frac{1}{\pi}\displaystyle \Re \int_{\C^+} \bar{\partial} \Phi_r(f)(z) h(z) d\ell_2(z),$$
where $\bar{\partial}=\frac{\partial}{\partial{x}}+\sqrt{-1}\frac{\partial}{\partial y}$, and $d\ell_2(z)$ refers to Lebesue measure on $\C^+$.  Below, we will write $\phi_f$ instead of $\phi_{f,r}$ to lighten notation, since the choice of $r$ will be clear from context.

A notable property of the functional $\phi_{f}$ is the so-called Helffer--Sj\"ostrand formula~\citep{Helffer:Sjostrand:1989}. For a suitable cut-off function $\chi$, this formula allows a generic linear spectral statistic $T_n(f)$ with $f\in\mathscr{C}_c^{r+1}(\R)$ to be represented in terms of the empirical Stieltjes transform $\hat{m}_n$,
$$T_n(f) = \phi_f(\hat{m}_n).$$
The importance of this formula in studying the fluctuations of eigenvalues has been recognized in several previous works; see for example \cite{Najim:Yao:2016} and the references therein. 

To describe the standardized statistic $p[T_n(\f)-E\{T_n(\f)\}]$, it will be convenient to define 
 the vector-valued functional $\phi_{\f}(h)=\{\phi_{f_1}(h),\dots,\phi_{f_m}(h)\}$, as well the following standardized versions of the Stieltjes transforms $\hat{m}_n(z)$ and  $\hat{m}_n^*(z;X)$,
 \begin{equation}\label{EQ:MUZDEF}
 \begin{split}
 \hat{\mu}_n(z) &=p[\hat{m}_n(z)-E\{\hat{m}_n(z)\}],\\[0.2cm]
 \hat{\mu}_n^*(z;X)&=p[\hat{m}_n^*(z;X)-E\{\hat{m}_n^*(z;X) \mid X\} ].
 \end{split}
 \end{equation}
 Consequently, linearity of the functional $\phi_{\f}$ implies the relations
\begin{equation*}
\begin{split}
\phi_{\f}(\hat{\mu}_n) &=p[T_n(\f)-E\{T_n(\f)\}], \\[0.2cm]
\phi_{\f}(\hat{\mu}_n^*) &= p[T_n^*(\f)-E\{T_n^*(\f) \mid X\} ].
\end{split}
\end{equation*}
%
In this notation, Theorem~\ref{THM:BOOTSTRAP_CONSIST} amounts to comparing the distributions $\mathcal{L}\{\phi_{\f}(\hat{\mu}_n)\}$ and $\mathcal{L}\{\phi_{\f}(\hat{\mu}_n^*) \mid X\}$ in the LP metric. To carry this out, each of these distributions will be compared separately with Gaussian processes evaluated under $\phi_{\f}$.
Specifically, let $G_n(z)$ denote the Gaussian process to be compared with $\hat{\mu}_n(z)$, and similarly, for a fixed realization of $X$, let $G_n^*(z;X)$ denote the Gaussian process to be compared with $\hat{\mu}_p^*(z;X)$. These processes will be defined precisely in Section~\ref{sec:Gaussian}. 
In turn, consider the bound
\begin{equation}\label{EQ:3TRIANGLE}
d_{\text{LP}}\Big[\mathcal{L}\big\{\phi_{\f}(\hat{\mu}_n)\big\}\, ,\,\mathcal{L}\{\phi_{\f}(\hat{\mu}_n^*) \mid X\big\}\Big] \leq \text{I}_n+\text{II}_n(X)+\text{III}_n(X),
\end{equation}
where we define the terms
\begin{align*}
\text{I}_n &= \;d_{\text{LP}}\Big[ \mathcal{L}\big\{\phi_{\f}(\hat{\mu}_n)\big\} ,\,\mathcal{L}\big\{\phi_{\f}(G_n)\big\}\Big],\\[0.2cm]
\text{II}_n(X)& = d_{\text{LP}}\Big[\mathcal{L}\big\{\phi_{\f}(G_n)\big\}\, , \, \mathcal{L}\big\{\phi_{\f}(G_n^*) \mid X\big\}\Big],   \\[0.2cm]
\text{III}_n(X) &  = d_{\text{LP}}\Big[ \mathcal{L}\big\{\phi_{\f}(G_n^*) \mid X\big\} \, , \,   \mathcal{L}\big\{\phi_{\f}(\hat{\mu}_n^*)\mid X\big\}\Big].
\end{align*}
It remains to show that $\text{I}_n+\text{II}_n(X)+\text{III}_n(X)$ converges to 0 in probability, which is handled in Section~\ref{sec:finish}.

\subsection{Defining the Gaussian processes $G_n(z)$ and $G_n^*(z;X)$}\label{sec:Gaussian}
Several pieces of notation will be needed to define the processes $G_n(z)$ and $G_n^*(z;X)$. First, let $b_0$ be any constant strictly greater than $(1+\sqrt{\gamma})\sup_{n}\lambda_1(\Sigma_n)$, and define the following domain in $\C^+$,
\begin{equation*}
\begin{split}
D^+&=[0,b_0]+\sqrt{-1}(0,1],
\end{split} 
\end{equation*}
as well as the symmetrized version 
$$D_{\text{sym}}=D^+\cup \overline{D^+},$$
 where the domain $\overline{D^+}$ consists of the complex conjugates of the points in $D^+$. For future reference, it is convenient to define
\begin{equation*}
D_{\ve}=[0,b_0]+\sqrt{-1}(\ve,1],
\end{equation*}
where $\ve\in(0,1)$ is fixed.

Next, recall that $t_n(z)$ denotes the Stieltjes transform of the distribution $\mathcal{F}(H_n,\gamma_n)$. 
It will also be convenient to use a modified version of $t_n(z)$, denoted
\begin{equation*}
\underline{t}_n(z)=-\ts\frac{1-\gamma_n}{z}+\gamma_nt_n(z).
\end{equation*}
We define bootstrap analogues of $t_n(z)$ and $\underline{t}_n(z)$, which should be viewed conditionally on a realization of the matrix $X$. Specifically, for a given realization of the estimator $\tilde\Lambda_n$ obtained from $X$, we define $\hat{t}_n(z)$ as the Stieltjes transform of the distribution $\mathcal{F}(\tilde H_n,\gamma_n)$. 
Likewise, we define
\begin{equation}\label{EQ:TUNDERHAT}
\underline{\hat{t}}_n(z)=-\ts\frac{1-\gamma_n}{z}+\gamma_n \hat{t}_n(z).
\end{equation}
We are now in position to define some parameters needed for describing the processes $G_n(z)$ and $G_n^*(z;X)$. Letting the complex derivative of $\underline{t}_n(z)$ be written as $\underline{t}'_n(z)$, and letting $z_1,z_2\in D_{\text{sym}}$, define the functions
\begin{align}
\Theta_{n,0}(z_1,z_2)&= \frac{\underline{t}_n'(z_1)\underline{t}_n'(z_2)}{\{\underline{t}_n(z_1)-\underline{t}_n(z_2)\}^2}-\frac{1}{(z_1-z_2)^2},\label{EQ:THETA0DEF}\\[0.3cm]
\Theta_{n,2}(z_1,z_2)&=\ts\frac{z_1^2z_2^2 \underline{t}_n'(z_1)\underline{t}_n'(z_2)}{n} \cdot \displaystyle\sum_{j=1}^p \{U_n\Gamma_{n,2}(z_1)U_n\ttop\}_{jj}\{U_n\Gamma_{n,2}(z_2)U_n\ttop\}_{jj},\label{EQ:THETA2DEFf}
\end{align}
where $\Gamma_{n,2}$ was defined in line~\eqref{EQ:GAMMADEF}. The above notation is drawn from the paper~\cite{Najim:Yao:2016}, and we omit another function $\Theta_{n,1}$ defined there, since it matches $\Theta_{n,0}$ in the context of real-valued data. The counterpart of $\Theta_{n,0}(z_1,z_2)$ in the bootstrap world is denoted $\hat{\Theta}_{n,0}(z_1,z_2)$, and is defined for each \smash{$z_1,z_2\in D_{\text{sym}}$} by replacing $\underline{t}_n$ with $\underline{\hat{t}}_n$. Meanwhile, the counterpart of $\Theta_{n,2}(z_1,z_2)$ in the bootstrap world is defined in terms of the matrix
\begin{equation}\label{EQ:GAMMAHATDEF}
\hat{\Gamma}_{n,\ell}(z)=\tilde\Lambda_n^{1/2}\Big[-zI_p +\big\{(1-\gamma_n)-z\gamma_n \hat{t}_n(z;X)\big\}\tilde\Lambda_n\Big]^{\!-\ell}\, \tilde \Lambda_n^{1/2}
\end{equation}
where $\ell\in\{1,2\}$,
and specifically
\begin{equation*}
 \hat{\Theta}_{n,2}(z_1,z_2;X)=\ts\frac{z_1^2z_2^2 \,\underline{\hat{t}}_n'(z_1)\underline{\hat{t}}_n'(z_2)}{n} \cdot \displaystyle\sum_{j=1}^p \{ \hat{\Gamma}_{n,2}(z_1;X)\}_{jj}\{\hat{\Gamma}_{n,2}(z_2;X)\}_{jj}.
\end{equation*}
With this notation in place, the following lemma defines the processes $G_n(z)$ and $G_n^*(z;X)$, and can be obtained as a reformulation of Proposition 5.2 in~\cite{Najim:Yao:2016}. Also, as a small clarification, for a generic complex-valued stochastic process indexed by $z$, say $W(z)\in\C$, we write its ordinary covariance function using the notation
$\text{cov}\{W(z_1),W(z_2)\}=E\big[(W(z_1)-E[W(z_1)])\,(W(z_2)-E[W(z_2))\big].$

\begin{lemma}\label{lem:gaussianexist}
Suppose that Assumptions~\ref{assumption:data} and~\ref{assumption:esd} hold. Then, for each $n\geq 1$, there exists a zero-mean complex-valued continuous Gaussian process $\{G_n(z)\}_{z\in D_{\text{\emph{sym}}}}$ with the covariance function
\begin{equation}\label{EQ:COVFORMULA1}
\text{\emph{cov}}\{G_n(z_1),G_n(z_2)\}=2\Theta_{n,0}(z_1,z_2)+(\kappa-3) \Theta_{n,2}(z_1,z_2).
\end{equation}
Also, for each $n\geq 1$, and almost every realization of $X$, there exists a zero-mean complex-valued continuous Gaussian process
$\{G_n^*(z;X)\}_{z\in D_{\text{\emph{sym}}}}$ with the conditional covariance function
\begin{equation}\label{EQ:COVFORMULA2}
\text{\emph{cov}}\{G_n^*(z_1;X),G_n^*(z_2;X) \mid X\}=2\hat{\Theta}_{n,0}(z_1,z_2;X)+(\hat{\kappa}_n-3) \hat{\Theta}_{n,2}(z_1,z_2;X),
\end{equation}
where $\hat{\kappa}_n$ is the kurtosis estimator~\eqref{EQ:EST-KURT} obtained from $X$. 
\end{lemma}

\subsection{Completing the proof of Theorem~\ref{THM:BOOTSTRAP_CONSIST}}\label{sec:finish}
Under the assumptions of Theorem~\ref{THM:BOOTSTRAP_CONSIST}, we will show that $\tI_n+\tII_n(X)+\tIII_n(X)$ converges to 0 in probability by applying the following core lemma from~\cite{Najim:Yao:2016} to each of the three terms separately. The details of applying the lemma are somewhat different for each term, and these details are addressed in separate paragraphs below. 
\begin{lemma}{\citep[Lemma 6.3]{Najim:Yao:2016}}. 
Let  $\{\hat{\varphi}_n(z)\}_{n\geq 1}$ and $\{\hat{\psi}_n(z)\}_{n\geq 1}$ be two sequences of centered complex-valued continuous stochastic processes indexed by $z\in D_{\text{\emph{sym}}}$. Assume the following conditions (i)--(v) hold.
\begin{enumerate}[(i)]
\item  For every $n\geq1$, and each $z\in D_{\text{sym}}$, the processes satisfy $\hat \varphi_n(\bar{z})=\overline{\hat\varphi_n(z)}$ and $\hat\psi_n(\bar{z})=\overline{\hat\psi_n(z)}$. 
\item  For every $\ve\in(0,1)$, both  sequences of processes $\{\hat\varphi_n(z)\}_{n\geq 1}$ and $\{\hat\psi_n(z)\}_{n\geq 1}$ are tight on $D_{\ve}$.
\item For every $n\geq 1$, the process $\hat\psi_n(z)$ is a complex-valued Gaussian process on $D_{\text{\emph{sym}}}$.
\item There are polynomial functions $\pi_{1}$ and $\pi_{2}$, not depending on $n$, such that the following bounds hold for every $n\geq 1$ and $z\in D^+$,
\begin{equation*}
\text{\emph{var}}(\hat\varphi_n(z))\leq \frac{\pi_{1}(|z|)}{\Im(z)^4}, \ \ \ \  \ \ \ \  \text{\emph{var}}(\hat\psi_n(z))\leq \frac{\pi_{2}(|z|)}{\Im(z)^4}.
\end{equation*}
\item For every fixed $d\geq 1$, and every $\{z_1,\dots,z_d\}\subset D^+$, the finite-dimensional distributions of $\{\hat\varphi_n(z)\}_{n\geq 1}$ and $\{\hat\psi_n(z)\}_{n\geq 1}$ satisfy
\begin{equation*}
d_{\text{\emph{LP}}}\big[\mathcal{L}\{\hat\varphi_n(z_1),\dots,\hat\varphi_n(z_d) \} \, , \, \mathcal{L}\{ \hat\psi_n(z_1),\dots,\hat\psi_n(z_d) \}\big]\to 0,
\end{equation*}
as  $n \to\infty$.
\end{enumerate}
Then, for any fixed collection of functions $\f=(f_1,\dots,f_m)$ lying in $\mathscr{C}_c^3(\R)$,
\begin{equation*}
d_{\text{\emph{LP}}}\big[\mathcal{L}\{\phi_{\f}( \hat\varphi_n)\} \, , \, \mathcal{L}\{\phi_{\f}(\hat\psi_n)\}\big] \to 0, 
\end{equation*}
as $n\to\infty$.
\end{lemma}
~\\
\noindent \emph{The term $\tI_n$.}
Consider the choices $\hat\varphi_n(z)=\hat{\mu}_n(z)$ and $\hat\psi_n(z)=G_n(z)$. Recall the definition $\hat\mu_n(z)=p[\hat m_n(z)-E\{\hat m_n(z)\}]$ from line~\eqref{EQ:MUZDEF}. Due to the fact that any Stieltjes transform $s(z)$ satisfies $\overline{s(z)}=s(\bar{z})$, it follows that property (i) holds for $\{\hat\mu_n(z)\}_{n\geq 1}$. Using this property of Stieltjes transforms again, the sequence $\{G_n(z)\}_{n\geq 1}$ can be verified to satisfy (i) by applying the meta-model argument in the proof of Proposition 5.2 in the paper~\cite{Najim:Yao:2016}, which implies that for each $n$, the process $G_n(z)$ arises as a limit of Stieltjes transforms. Our notation for $G_n(z)$ differs from that in~\cite{Najim:Yao:2016}, since the process $G_n(z)$ has mean zero here. Lastly, the fact that both sequences of processes satisfy conditions (ii)-(v) follow directly from Theorem 1 and Proposition 6.4 in~\cite{Najim:Yao:2016}; see also the comments preceding Proposition 6.4. Therefore, $\tI_n\to 0$.\\

\emph{Remark.}
\normalfont
In the remaining paragraphs, we will write $X_n$ instead of $X$, in order to emphasize the fact that each realization of $X$ lies within the sequence of matrices $\{X_n\}_{n\geq 1}$.\\

\noindent \emph{The term $\tII_n(X_n)$.}
Consider the choices $\hat\varphi_n(z)=G_n(z)$ and $\hat\psi_n(z)=G_n^*(z;X_n)$, where we view the second process from the viewpoint of the bootstrap world, conditionally on a fixed realization of $X_n$. In the previous paragraph, we already explained why $\{G_n(z)\}_{n\geq 1}$ satisfies conditions (i)--(iv). Hence, we will first verify the conditions (i)--(iv) for $\{G_n^*(z;X_n)\}_{n\geq 1}$, and then condition (v) involving both sequences of processes will be verified later. 
To handle the first task, it is enough to show that for any subsequence $\mathcal{N}\subset\{1,2,\dots\}$, there is a sub-subsequence $\mathcal{N}'\subset \mathcal{N}$, such that $\{G_{n}^*(z;X_{n})\}_{n\in \mathcal{N}'}$ satisfies conditions (i)-(iv) for almost every realization of $\{X_{n}\}_{n\in \mathcal{N}'}$.
For the conditions (i)--(iv), the arguments used in the previous paragraph may be applied almost directly to $\{G_{n}^*(z;X_{n})\}_{n\in \mathcal{N}'}$, again using Theorem 1 and Proposition 6.4 in~\cite{Najim:Yao:2016}. However, there is one detail to notice, which is that if we view $\{G_{n}^*(z;X_{n})\}_{n\in \mathcal{N}'}$ from the perspective of the bootstrap world, then the population kurtosis is $\hat\kappa_n$, which varies with $n$, whereas $\kappa$ is  fixed with respect to $n$. Nevertheless, this does not create any difficulty when using Theorem 1 and Proposition 6.4 in~\cite{Najim:Yao:2016}. The reason is that the proofs underlying these results allow the 
population kurtosis to vary with $n$ as long as it remains bounded, and since we know $\hat\kappa_n\xrightarrow{ \ \P \ } \kappa$, it follows that we can find a sub-subsequence $\mathcal{N}'$ such that almost every realization of $\{\hat\kappa_n\}_{n\in\mathcal{N}'}$ is bounded.

We now verify condition (v) almost surely along subsequences by showing that for any fixed set $\{z_1,\dots,z_d\}\subset D^+$,  the following limit holds
as $n\to\infty$, 
\begin{equation*}
d_{\text{LP}}\Big[\mathcal{L}\big\{G_n(z_1),\dots,G_n(z_d)\big\} \, , \, \mathcal{L}\big\{G_n^*(z_1;X_n),\dots,G_n^*(z_d;X_n) \mid X_n\big\}\Big] \ \ \xrightarrow{ \ \ \P \ \ } \ \ 0.
\end{equation*}
Here, it is important to keep in mind that  $G_n(z)$ and $G_n^*(z;X_n)$ are centered complex-valued Gaussian processes. Unlike the case of real-valued Gaussian processes, there is a small subtlety, because in general, if $W(z)$, say, is a centered complex-valued Gaussian process, then its finite dimensional distributions depend on both the ordinary covariance function $E\{W(z_1)W(z_2)\}$, as well as the conjugated version $E \{W(z_1)\overline{W(z_2)} \}$. However, since the processes $G_n(z)$ and $G_n^*(z;X_n)$ satisfy condition (i), and since the domain $D_{\text{sym}}$ is closed under complex conjugation, the finite-dimensional distributions of $G_n(z)$ and $G_n^*(z;X_n)$ are completely determined by their ordinary covariance functions on $D_{\text{sym}}^2$.

Due to the comments just given, the task of verifying (v) reduces to showing that there is a limiting covariance function $C(z_1,z_2)$  such that the following limits hold for all $(z_1,z_2)\in D_{\text{sym}}^2$,
 \begin{equation}\label{EQ:COVLIMIT}
 \text{cov}\{G_n(z),G_n(z_2)\} \xrightarrow{ \ \ } C(z_1,z_2),
 \end{equation}
 and
  \begin{equation}\label{EQ:COVLIMITSTAR}
 \text{cov}\{G_n^*(z;X_n),G_n^*(z_2;X_n) \mid X_n \} \xrightarrow{ \ \P \ } C(z_1,z_2).
 \end{equation}
 By inspecting the covariance formulas~\eqref{EQ:COVFORMULA1} and~\eqref{EQ:COVFORMULA2}, and using $\hat\kappa_n\xrightarrow{\P} \kappa $, it follows that the above limits~\eqref{EQ:COVLIMIT} and~\eqref{EQ:COVLIMITSTAR} will hold if we can show that for each $\ell\in\{0,2\}$, there is a deterministic function $\Theta_{\ell}(z_1,z_2)$ on  the domain $D_{\text{sym}}^2$ such that 
 \begin{equation}\label{EQ:THETALIM}
 \Theta_{\ell,n}(z_1,z_2) \xrightarrow{ \ \ } \Theta_{\ell}(z_1,z_2),
 \end{equation}
 and
 \begin{equation}\label{EQ:THETASTARLIM}
 \hat{\Theta}_{\ell,n}(z_1,z_2;X_n) \xrightarrow{ \P } \Theta_{\ell}(z_1,z_2).
 \end{equation}

To handle the limit~\eqref{EQ:THETALIM}, let $t(z)$ denote the Stieltjes transform of $\mathcal{F}(H,\gamma)$, and let $\underline{t}(z)=-\frac{1-\gamma}{z}+\gamma t(z)$. In the special situation when $\Sigma_n$ is diagonal for every $n\geq 1$, the calculations in~\citet[Section 3.4, see also p.\,1845]{Najim:Yao:2016} show that under Assumptions~\ref{assumption:data} and~\ref{assumption:esd}, the limit~\eqref{EQ:THETALIM} exists for each $\ell\in\{0,2\}$, and each function $\Theta_{\ell}(z_1,z_2)$ is determined by $\gamma$ and $H$.
When $\Sigma_n$ is not diagonal, Assumption~\ref{assumption:eigenvec} may be used, since it implies that $\Theta_{n,2}(z_1,z_2)$, as defined in~\eqref{EQ:THETA2DEFf}, still behaves asymptotically as if $\Sigma_n$ were diagonal. As a side note, observe that if $\kappa=3$, then the  formula~\eqref{EQ:COVFORMULA1} shows that  $\Theta_{n,2}(z_1,z_2)$  does not affect the limiting covariance function $C(z_1,z_2)$. This explains why Assumption~\ref{assumption:eigenvec} is not needed when $\kappa=3$.

Next, we handle the bootstrap limit~\eqref{EQ:THETASTARLIM}, and in fact, we show that it holds almost surely. The idea is to check that the same conditions giving rise to $\Theta_{\ell}(z_1,z_2)$ in the limit~\eqref{EQ:THETALIM} also hold in the bootstrap world.  Specifically, the calculations in \citet[Section 3.4]{Najim:Yao:2016} that establish the limit~\eqref{EQ:THETALIM} are based on four conditions: that $\Sigma_n$ is diagonal, that $\sup_{n}\lambda_1(\Sigma_n)<\infty$, that $H_n\Rightarrow H$, and that $\gamma_n\to \gamma$. In light of these conditions, we proceed by viewing  $\tilde{\Lambda}_n$ as a diagonal population covariance matrix in the bootstrap world, and by viewing $\tilde{H}_n$ as  the analogue of $H_n$ in the bootstrap world. It follows from our Theorem~\ref{THM:CONSIST} that for almost every realization of the matrices $\{X_n\}_{n\geq 1}$, the conditions $\sup_n \lambda_1(\tilde\Lambda_n)<\infty$ and $\tilde H_n\Rightarrow H$ are satisfied. Therefore, we conclude that for each $\ell\in\{0,2\}$, the limit  $\hat{\Theta}_{\ell,n}(z_1,z_2;X)\to \Theta_{\ell}(z_1,z_2)$ holds almost surely. The completes the verification of the limit $\tII_n(X)\xrightarrow{ \ \P \ }0$.\\

\noindent \emph{The term $\tIII_n(X_n)$.} Consider the choices $\hat\varphi_n(z)=\hat{\mu}_n^*(z;X_n)$ and $\hat\psi_n(z)=G_n^*(z;X_n)$. The conditions (i)-(v) can be verified using the same reasoning described for $\tI$ and $\tII_n(X_n)$ above. It follows that \smash{$\tIII_n(X_n)\xrightarrow{ \P } 0$.}

\section{Proof of Theorem~\ref{THM:BOOTSTRAP_BIAS}}\label{sec:biasproof}
\subsection{The limit~\eqref{EQ:BIASMETRIC}} Here we explain how the second limit~\eqref{EQ:BIASMETRIC} follows quickly from the first limit~\eqref{EQ:BIASRESULT}, in conjunction with Theorem~\ref{THM:BOOTSTRAP_CONSIST}. For ease of notation we define the random vectors
\begin{equation}
\begin{split}
U_n &= p[T_n(\f)-E\{T_n(\f)\}]\\[0.2cm]
U_n^* &=p[T_{n,1}^*(\f)-E\{T_{n,1}^*(\f) \mid X\}].
\end{split}
\end{equation}
By the triangle inequality, 
\begin{equation}
 \begin{split}
 d_{\rm LP}\Big[\mathcal{L}\big\{p\{T_n(\f)-\vartheta_n(\f)\}\big\} \ , \ \mathcal{L}\big\{p\{T_{n,1}^*(\f)-\tilde \vartheta_n(\f)\} \mid X\big\}\Big] & \ \leq \ \Delta_{n,1}(X)+\Delta_{n,2}(X),
 \end{split}
\end{equation} 
where
\begin{equation}
\begin{split}
\Delta_{n,1}(X) &=d_{\text{LP}}\Big[\mathcal{L}\{U_n+pb_n(\f)\}\, , \, \mathcal{L}\{U_n^*+pb_n(\f)\mid X \}\Big] \\[0.2cm]
\Delta_{n,2}(X)&=d_{\text{LP}}\Big[\mathcal{L}\{U_n^*+pb_n(\f) \mid X\} \, , \, \mathcal{L}\{ U_n^*+p\hat b_n(\f) \mid X\}\Big].
\end{split}
\end{equation}
Due to the translation-invariance of the LP metric, Theorem~\ref{THM:BOOTSTRAP_CONSIST} implies $\Delta_{n,1}(X)\xrightarrow{ \ \P \ } 0$.  To handle $\Delta_{n,2}(X)$, it is a basic fact that if two random vectors are related by a constant translation, then the LP distance between them is at most the length of the translation. Therefore,
$$\Delta_{n,2}(X) \ \leq \ p\|b_n(\f)- \hat b_n(\f)\|_2,$$
and this bound tends to 0 in probability by the first limit~\eqref{EQ:BIASRESULT}.\qed

\subsection{The limit~\eqref{EQ:BIASRESULT}}
The proof is decomposed into two results below, Propositions~\ref{prop:bias1} and~\ref{prop:bias2}, which directly imply the limit~\eqref{EQ:BIASRESULT}. Before stating these results, a fair bit of notation is needed. The first proposition shows that the bias $b_n(\f)$ is asymptotically equivalent to another vector $\mathfrak{b}_n(\f)$, in the sense that $p\{b_n(\f)-\mathfrak{b}_n(\f)\}\to 0$.  When the components of $\f=(f_1,\dots,f_m)$ lie in
$\mathscr{C}_c^{18}(\R)$, the vector $\mathfrak{b}_n(\f)$ is defined by
\begin{equation}\label{EQ:FRAKBDEF}
\mathfrak{b}_n(\f)=\frac{1}{\pi p}\Re\int_{\C^+}\bar\partial \Phi_{17}(\f)(z)\mathcal{B}_n(z)d\ell_2(z),
\end{equation}
where for any $z\in\C^+$, the function $\B_n(z)$ is set to
\begin{equation}\label{EQ:CALBDEF}
\B_n(z)=\B_{n,1}(z)+(\kappa -3)\B_{n,2}(z),
\end{equation}
with the terms being defined as
\begin{align}
 \B_{n,1}(z) &=\frac{-z^3\underline{t}_n^3(z)\ts\frac 1n \tr\big\{ \Gamma_{n,2}(z)\Gamma_{n,1}(z)\big\}}{\Big[1-z^2\underline{t}_n^2(z)\ts\frac{1}{n}\tr\big\{\Lambda_n \Gamma_{n,2}(z)\big\}\Big]\cdot\Big[1-z^2\underline{t}_n^2(z)\ts\frac{1}{n}\tr\big\{\Gamma_{n,1}^2(z)\big\}\Big]}\label{EQ:CALB1DEF},\\[0.5cm]
   \B_{n,2}(z)&=-z^3\underline{t}_n^3(z)\frac{\ts\frac{1}{n}\sum_{j=1}^p \{U_n\Gamma_{n,1}(z)U_n\ttop\}_{jj}\{U_n\Gamma_{n,2}(z)U_n\ttop\}_{jj}}{1-z^2\underline{t}_n^2(z)\ts\frac{1}{n}\tr\{\Lambda_n\Gamma_{n,2}(z)\}}\label{EQ:CALB2DEF}.
   \end{align}
Proposition~\ref{prop:bias1} also shows that the bias estimate $\hat b_n(\f)$ is asymptotically equivalent to a vector  $\hat{\mathfrak{b}}_n(\f)$ given by
\begin{equation}\label{EQ:BFRAKHATDEF}
 \hat{\mathfrak{b}}_n(\f)=\frac{1}{\pi p}\Re\int_{\C^+}\bar\partial \Phi_{17}(\f)(z)\hat{\mathcal{B}}_n(z)d\ell_2(z),
\end{equation}
where
\begin{equation}\label{EQ:BHATDEF}
 \hat{\B}_n(z)=\hat{\B}_{n,1}(z)+(\hat{\kappa}_n-3) \hat{\B}_{n,2}(z),
\end{equation}
and the terms $\hat{\B}_{n,1}(z)$ and $\hat{\B}_{n,2}(z)$ are defined in analogy with~\eqref{EQ:CALB1DEF} and~\eqref{EQ:CALB2DEF}. Specifically, the terms $\hat{\B}_{n,1}(z)$ and $\hat{\B}_{n,2}(z)$ are defined by replacing $\underline{t}_n(z)$, $\Gamma_{n,\ell}(z)$, and $\Lambda_n$ respectively with the counterparts $\hat{\underline{t}}_n(z)$, $\hat{\Gamma}_{n,\ell}(z)$, and $\tilde\Lambda_n$.
In addition, the matrix $U_n$ in the formula for $\hat{\B}_{n,2}(z)$ is  replaced with $I_p$.
   %
   
   
   All of the notation for this section is now in place, and so we may state the following result, which is an adaptation of~Theorem 3 in \citep{Najim:Yao:2016}.


\begin{proposition}\label{prop:bias1}
Suppose that Assumptions~\ref{assumption:data} and~\ref{assumption:esd} hold. Let $\f=(f_1,\dots,f_m)$ be fixed functions lying in $\mathscr{C}_c^{18}(\R)$.  Then, as $n\to\infty$,
\begin{equation}\label{EQ:BFRAKLIM}
 p\{b_n(\f)-\mathfrak{b}_n(\f)\}\xrightarrow{ \ \ } 0.
\end{equation}
and
\begin{equation}\label{EQ:BOOTFRAKLIM}
 p\{\hat{b}_n(\f)-\hat{\mathfrak{b}}_n(\f)\}\xrightarrow{ \ \P \ } 0.
\end{equation}
\end{proposition}
%
%

\begin{proof} The first limit~\eqref{EQ:BFRAKLIM} is the conclusion of~\citep[Theorem 3]{Najim:Yao:2016}, and the assumptions for that result are immediately implied by Assumptions~\ref{assumption:data} and~\ref{assumption:esd} here. To obtain the second limit~\eqref{EQ:BOOTFRAKLIM}, the proof of~\citep[Theorem 3]{Najim:Yao:2016} can be carried out analogously in the bootstrap world. In particular, the conditions used in that proof hold almost surely along subsequences, since $\hat\kappa_n\xrightarrow{\P} \kappa$, and $\sup_n\lambda_1(\tilde \Lambda_n)<\infty$ almost surely. 
\end{proof}


\begin{proposition}\label{prop:bias2}
Suppose that Assumptions~\ref{assumption:data} and~\ref{assumption:esd} hold, and that either $\kappa=3$ or Assumption~\ref{assumption:eigenvec} hold. Let $\f=(f_1,\dots,f_m)$ be fixed functions lying in $\mathscr{C}_c^{8}(\R)$.  Then, there exists a fixed vector $\beta(\f)\in\R^m$ such that as $n\to\infty$,
\begin{equation}\label{EQ:BETALIM1}
 p\mathfrak{b}_n(\f)\xrightarrow{ \  \ } \beta(\f),
\end{equation}
and
\begin{equation}\label{EQ:BETALIM2}
 p\hat{\mathfrak{b}}_n(\f)\xrightarrow{ \ \P  \ } \beta(\f)
\end{equation}
\end{proposition}

\begin{proof} Under the spectrum regularity condition (Assumption~\ref{assumption:esd}), the calculations on p.1852 of~\citep{Najim:Yao:2016} show that as $n\to\infty$, the quantity $\B_{n,1}(z)$ converges to the following limit for any fixed $z\in\C^+$,
\begin{equation}
 \B_{n,1}(z) \xrightarrow{ \  \ } \B_1(z):=\frac{-\gamma z^3 \underline{t}^3(z)}{(1-\mathcal{A}(z))^2}\int\frac{\lambda^2}{(1+\lambda \underline{t}(z))^3}dH(\lambda),
\end{equation}
where we define
\begin{equation}
 \mathcal{A}(z)=\gamma \,\underline{t}^2(z)\int\frac{\lambda^2}{(1+\lambda\underline{t}(z))^2}\, dH(\lambda).
\end{equation}
Furthermore, if the eigenvector regularity condition,~Assumption~\ref{assumption:eigenvec}, also holds, then the same set of calculations~\citep[p.1852]{Najim:Yao:2016} gives the limit
\begin{equation}
\B_{n,2}(z)\xrightarrow{ \ \ } \B_2(z):=\{1-\mathcal{A}(z)\}\mathcal{B}_{1}(z).
\end{equation}
The reason that Assumption~\ref{assumption:eigenvec} is needed here is that it ensures that $\B_{n,2}(z)$ behaves as if $\Sigma_n$ is diagonal, which is required in the calculations just mentioned. Likewise, if we put $\B(z)=\B_1(z)+(\kappa-3)\B_2(z)$, then
$$\B_n(z)\to \B(z).$$
In the case $\kappa=3$, the term $\B_2(z)$ becomes irrelevant, and then Assumption~\ref{assumption:eigenvec} is no longer needed for handling the limit $\B_{n,2}(z)\to\B_2(z)$. 

To apply the work above, recall that $\mathfrak{b}_n(\f)=\frac{1}{\pi p}\Re\int_{\C^+}\bar\partial \Phi_{7}(\f)(z)\mathcal{B}_n(z)d\ell_2(z)$. Hence, if we define
\begin{equation}
\beta(\f)=\frac{1}{\pi}\Re \displaystyle\int_{\C^+}\bar\partial \Phi_{7}(\f)(z)\B(z)d\ell_2(z),
\end{equation}
then the dominated convergence theorem will give the claimed limit~\eqref{EQ:BETALIM1}, provided that $|\bar\partial \Phi_{7}(\f)(z)\B_n(z)|$ is dominated by a fixed integrable function on $\C^+$.
For this purpose, the proof of~Proposition 6.2 and the bound in line~7.4 of~\citep{Najim:Yao:2016} show that if the functions $\f=(f_1,\dots,f_m)$ lie in $\mathscr{C}_c^8(\R)$, then there is an integrable function $g$ on $\C^+$ such that
$$\sup_n |\bar\partial \Phi_{7}(\f)(z)\B_n(z)| \ \leq \ |g(z)|,$$
for every $z\in\C^+$.
This completes the proof of the limit~\eqref{EQ:BETALIM1}.

The proof of the limit~\eqref{EQ:BETALIM2} is largely similar, but with a few minor differences. The main points to notice are that in the bootstrap world, the diagonal matrix $\tilde \Lambda_n$ plays the role of the population covariance matrix, and the associated spectral distribution satisfies $\tilde H_n\Rightarrow H$ almost surely, by Theorem~\ref{THM:CONSIST}. Therefore, the calculations from~\citep[p.~1852]{Najim:Yao:2016} may be re-used to show that for each $z\in\C^+$, and each $l\in\{1,2\}$, the following limit holds
\begin{equation}
\hat \B_{n,l}(z) \to \B_l(z)  \ \ \text{ almost surely.}
\end{equation}
Likewise, since Theorem~\ref{THM:CONSIST} gives $\hat\kappa_n\xrightarrow{ \ \P \ } \kappa$, it follows that the quantity $\hat\B_n(z)$ (defined in line~\eqref{EQ:BHATDEF}) satisfies $\hat\B_n(z)\xrightarrow{ \ \P \ } \B(z)$. Furthermore, the previous dominated convergence argument for $|\bar\partial \Phi_{7}(\f)(z)\B_n(z)|$ can be essentially repeated for $|\bar\partial \Phi_{7}(\f)(z)\hat\B_n(z)|$, which leads to the limit~\eqref{EQ:BETALIM2}.
\end{proof}

\end{document}